\documentclass[twocolumn,journal]{IEEEtran}

\usepackage[latin1]{inputenc}
\usepackage{graphicx, amssymb, amsmath, amsfonts, amsthm}
\usepackage[usenames,dvipsnames]{color}
\usepackage{bbm,multirow,bm,hyperref,makecell}

\usepackage{algorithm}
\usepackage{algpseudocode}

\newtheorem{corollary}{Corollary}
\newtheorem{lemma}{Lemma}

\newtheorem{proposition}{Proposition}
\usepackage[short]{optidef}

\newcommand{\C}{\mathcal{C}}
 
\DeclareMathOperator{\diag}{diag}

\allowdisplaybreaks
\IEEEoverridecommandlockouts

\begin{document}
\title{$1$-bit RIS-aided Index Modulation with \\Quantum Annealing}

\author{Ioannis Krikidis, \IEEEmembership{Fellow, IEEE}, Constantinos Psomas, \IEEEmembership{Senior Member, IEEE}, and Gan Zheng, \IEEEmembership{Fellow, IEEE}\vspace*{-5mm}
\thanks{I. Krikidis is with the Department of Electrical and Computer Engineering, University of Cyprus, Cyprus (e-mail: krikidis@ucy.ac.cy).}
\thanks{C. Psomas is with the Department of Computer Science and Engineering, European University Cyprus, Cyprus (e-mail: c.psomas@euc.ac.cy).}
\thanks{G. Zheng is with the School of Engineering, University of Warwick, Coventry, CV4 7AL, UK (e-mail: gan.zheng@warwick.ac.uk).} \thanks{This work was supported by the European Research Council (ERC) under the Horizon Europe programme (Grant No. 101241675, PoC QUARTO), the EU Horizon-JU SNS programme (Grant No. 101192080, 6G-LEADER), and the Cyprus Research \& Innovation Foundation (Grant No. DUAL USE/0922/0031). Part of this work was also supported by the UK Engineering and Physical Sciences Research Council (EPSRC) grant EP/X04047X/2 for TITAN Telecoms Hub.}}

\maketitle

\begin{abstract}
In this paper, we investigate a new index modulation (IM) scheme for reconfigurable intelligent surface (RIS)-assisted communications with $1$-bit RIS phase resolution. In addition to the traditional modulated symbols, extra bits of information are embedded in the binary RIS phase vector by indexing the cardinality of the positive phases shifts. To maximize capacity, the IM-based RIS vector is selected so as to maximize the signal-to-noise ratio at the receiver. The proposed IM design requires the solution of a quadratic binary optimization problem with an equality constraint at the transmitter as well as a quadratic unconstrained binary optimization (QUBO) problem at the receiver. Since commercial solvers cannot directly handle constraints, a penalty method that embeds the equality constraint in the objective function is investigated. To overcome the empirical tuning of the penalty parameter, an iterative Augmented Lagrangian optimization technique is also investigated where a QUBO problem is solved at each iteration. The proposed design and associated mathematical framework are tested in a real-world quantum annealing device provided by D-WAVE. Rigorous experimental results demonstrate that the D-WAVE heuristic efficiently solves the considered combinatorial problems. Furthermore, theoretical bounds on the average capacity are provided. Both experimental and theoretical results show that the proposed design outperforms conventional counterparts.
\end{abstract}


\section{Introduction}

\IEEEPARstart{6}{G} communication systems introduce new engineering requirements such as data rates up to $1$ Tbps, extremely high connectivity, ultra reliability and energy efficiency, as well as air-interface latency $10\times$ lower  than current 5G deployments \cite{TAT}. To satisfy these requirements, industry and academia focus on new emerging technologies and communication paradigms such as, reconfigurable intelligent surfaces (RISs), THz frequency bands, integrated sensing and communications, advanced modulation techniques such as index modulation (IM), low-bit-resolution signal processing, etc. Specifically, RISs refer to intermediate passive, active or hybrid devices that scatter the incident radio signals intelligently to control the wireless channel \cite{WU}. By jointly adjusting their reflection coefficients, they can enhance the signal-to-noise ratio (SNR), achieving higher data rates and/or multi-user interference mitigation. Moreover, due to their quasi-passive nature and the reduced number of RF chains required for their operation, RISs are both cost-effective and energy-efficient \cite{JSTSP}. On the other hand, IM is a modulation technique that conveys information by varying the index of transmission, such as transmit antennas, subcarriers or precoding matrices, rather than altering the amplitude/phase/frequency of the transmitted sinewave as in conventional modulation schemes \cite{IM}. The additional information conveyed through the IM indices increases the spectral efficiency and provides robust communications in challenging environments. Furthermore, the simplicity of IM leads to reduced implementation and hardware complexity as well as improved overall energy efficiency. Finally, low-bit-resolution communication systems that utilize only one or a few bits for signal encoding, are introduced to balance the trade-off between performance and implementation complexity \cite{SIL,ALI}.

Due to the potential benefits arising from their implementation, both RIS and IM technologies have been extensively studied in the literature. See \cite{WU,IM} and references within. Hence, as expected, there are also works that investigate the synergy of these two technologies \cite{MDR}. In \cite{BAS}, the author proposes two schemes, namely, an RIS-aided space shift keying scheme and an RIS-aided spatial modulation scheme. The two schemes use the IM concept at the receiver's antennas and exploit the RIS to maximize the received signal at the selected antenna. The work in \cite{BAY} follows a similar approach with the consideration of Alamouti space-time code. The proposed technique transmits the coded information through the RIS to the targeted receive antenna. The IM principles have also been utilized at the RIS side. In particular, the work in \cite{MDR2} investigates a reflection pattern modulation, where the RIS activates a subset of its elements to convey additional information. A similar design is studied in \cite{SHEN} focusing on visible light communications. The authors of \cite{HAN} provide a technique with two RISs and propose a beam-IM scheme, ideal for millimeter wave communications.

The aforementioned works consider RISs with continuous phase shifts, {\it i.e.,} infinite resolution. However, RISs with low resolution phase shifts, {\it e.g.,} $1$-bit phase shifters, are more practical, more cost-effective and of lower complexity \cite{ZHANG}. To the best of our knowledge, low-bit resolution RIS communications with IM techniques has not been addressed yet in the literature. In fact, the design and optimization of IM-based RIS-aided communications with discrete phase shifts require a complete different approach, mainly due to the combinatorial nature of the resulting setup. What is more, all these technological breakthroughs considerably enhance computational demands and necessitate computing resources with exceptionally high capabilities. Unfortunately, traditional silicon-based Von Neumann computing architectures can not be further advanced since transistors have reached their atomic limits \cite{ITR}. Quantum computing emerges as a promising alternative to address this computational challenge, offering a suitable platform for advanced wireless technologies. The integration of quantum computing architectures and algorithms within wireless communication systems represents a crucial and emerging field of research \cite{KIM,HANZO}.

Quantum computing exploits quantum mechanics and relies on the principles of quantum tunneling, superposition and entanglement, encompassing two primary models: gate-based quantum computing and single-purpose quantum annealing (QA) model. The gate-based model operates discretely, utilizing programmable logic gates (unitary and reversible transformations) that act on qubits, resembling classical digital architectures. By sequentially interconnecting these basic logic gates, a variety of quantum algorithms {\it e.g.}, Grover search, Deutsch, Shor, quantum Fourier transform, quantum approximate optimization algorithm etc.,  can be implemented, often achieving a computational speedup compared to classical counterparts \cite{KAY}. For instance, the quadrature speed-up of the Grover's algorithm (quantum search in an unsorted database) has been already proposed in the literature to solve several physical layer problems in wireless communication systems \cite{BOT}. A first effort to combine IM with gate-based model quantum computing is presented in \cite{NI}; the proposed technique applies adaptive Grover search and simulation results are presented by using IBM Qiskit. However, gate-based quantum devices are highly susceptible to quantum decoherence, limiting the number of qubits and logic gates that can be effectively utilized. 

In contrast, the QA model is analog and grounded in the adiabatic principle of quantum mechanics (Adiabatic Theorem) \cite{MCG, YAR}. This model is particularly effective for addressing NP-hard combinatorial optimization problems, which are formulated as instances of the Ising model (spin-glass) which is equivalent to the quadratic unconstrained binary optimization (QUBO) problem. By leveraging the system's adiabatic evolution, it gradually evolves to a final Hamiltonian ({\it i.e.,} energy function representing the objective function to be minimized) whose ground state (lowest energy level) embeds the solution of the optimization problem considered. D-WAVE is a commercial analog quantum device that performs a noisy version of QA using advanced superconducting integrated circuits (niobium loops) \cite{WAVE}. Recently, it has gained significant attention for its impressive computational capabilities, including a high number of qubits, and its user-friendly cloud-based programming interface for real-time remote access. The current D-WAVE {\it Advantage} system features a quantum processing unit  with over $5,000$ flux qubits, while the forthcoming D-WAVE architectures are anticipated to incorporate more than $7,000$ flux qubits with larger qubit connectivity. Although there is not any theoretical result/proof about the performance/speed-up of the D-WAVE QA, experimental results show quantum advantages against classical heuristics for particular instances. D-WAVE QA has been used to various traditional NP-hard problems ({\it e.g.,} set cover problem, Knapsack, portfolio optimization, etc.) and experimental results show significant performance benefits \cite{FRO, KOS, DJI, TAS}. 

QA has been utilized to solve a wide range of problems in wireless communications. Early works mainly focus on addressing the maximum-likelihood signal detection problem in large-scale multiple-input multiple-output (MIMO) systems \cite{JAM1}. Experimental results show significant performance benefits in comparison to conventional solutions as well as limitations due to the number of qubits available and the noise effects (integrated control errors) \cite{ZABO}. QA has been also used to address other design problems of combinatorial nature in wireless communication systems such as beam assignment for satellite systems \cite{DIN}, design of RIS phases shifts in RIS-assisted communications \cite{LIM}, multi-user detection for non-orthogonal multiple access \cite{YON}, decoding for polar codes \cite{KASI}, antenna configuration selection in fluid-antenna MIMO systems \cite{KRI1}, etc. In our recent work, D-WAVE QA has been also used to design pre/post coding vectors for MIMO point-to-point MIMO systems with $1$-bit analogue \cite{KRI} and digital resolution \cite{KRI2}, respectively. The results show that D-WAVE QA is a promising solver/heuristic for low-bit-resolution signal processing techniques that are represented by QUBO instances. It is also worth noting that QA has been employed as an essential block to speed-up the training process in machine-learning models, which has potential interest in the design of wireless communication systems \cite{AMI,DAT}.   

In contrast to the discussed literature, this paper investigates a novel $1$-bit RIS signal processing technique, which merges RIS-aided communication with IM. Specifically, the new IM design embeds information bits in the binary RIS phase vectors by indexing the cardinality of their distinct phases shifts. Then, the RIS chooses the vector that also maximizes the received SNR. The proposed IM scheme requires the solution of a binary optimization problem with an equality constraint at the transmitter and a QUBO problem at the receiver, which are solved in a real-world D-WAVE QA machine by using an appropriate mathematical framework. Specifically, the major contributions of the paper are summarized as follows: 
\begin{itemize}
\item A new RIS-based IM technique is proposed in which index information is embedded in the phase beamforming vectors that are used by the RIS. By taking into account the binary resolution $\{+1,-1\}$ of the phase shifts, additional information bits are conveyed by using the cardinality of $\pm 1$ entries in the RIS vectors as IM information. To maximize capacity, the RIS vector (with a given number of $\pm 1$ phase shifts) that maximizes the SNR is selected for the IM process. Numerical results show that the proposed scheme outperforms conventional passive beamforming with binary phase shifts, while the performance advantages increase as the number of RIS elements increases.    
\item The proposed RIS-based IM scheme is represented by a combinatorial binary problem with a single equality constraint at the transmitter side, as well as a conventional QUBO problem at the receiver side. To tackle the equality constraint, we firstly study the penalty method which embeds the constraint into objective function by using an auxiliary penalty parameter. An iterative technique that iterates over different penalty parameters is also considered. To avoid the empirical tuning of the penalty parameter, we investigate the Augmented Lagrangian (AL) method, which combines Lagrangian optimization with the penalty method; this iterative method solves a QUBO problem at each iteration, while it converges to the solution faster and in a more smooth way due to the additional free variable.
\item Theoretical expressions are derived for the average SNR achieved by the proposed RIS-based IM scheme, characterizing how the cardinality of $\pm 1$ entries in the RIS vector affects the received signal power. The proposed scheme performs close to conventional beamforming (in terms of SNR) when the numbers of $+1$ and $-1$ phase shifts are equal (or approximately equal) and, for a sufficiently large number of RIS elements, the two schemes have similar performance. Moreover, we derive an upper bound on the average capacity. Our results show that the RIS-based IM technique outperforms conventional beamforming in terms of average capacity.
\item To solve the considered QUBO problems, we employ a real-world state-of-the art QA device {\it i.e.,} D-WAVE {\it Advantage\_System6.4}. Rigorous experimental results show that the proposed RIS-based IM design can be efficiently solved through QA. Due to the low resolution of the quantum hardware, the AL technique seems to outperform the penalty method for large RIS configurations, since it handles the equality constraint in a more numerically robust manner.    
\end{itemize}

\noindent {\it Notation:} Lower and upper case bold symbols denote vectors and matrices, respectively; the superscripts $(\pmb{x})^{T}$, $(\pmb{x})^{H}$ denote transpose and conjugate transpose of the vector $\pmb{x}$, respectively; $\mathbbm{E}(\cdot)$ denotes the statistical expectation; $\pmb{1}$ is a full-ones column vector of appropriate dimension, $\Re(\cdot)$ returns the real part of its complex argument, $\diag(\pmb{x})$ denotes a diagonal matrix whose main diagonal is $\pmb{x}$, $\pmb{I}$ is the identity matrix of appropriate dimension, $\mathbb{C}^{N\times M}$ denotes the set of complex-valued $N\times M$ matrices, and $\mathcal{CN}(\mu,\sigma^2)$ represents the complex Gaussian distribution with mean $\mu$ and variance $\sigma^2$.

The remainder of the paper is organized as follows: Section \ref{sec2} introduces the system model and provides some useful background. Section \ref{sec3} presents the proposed RIS-aided IM scheme with the associated optimization problems and theoretical framework, whereas Section \ref{sec4} focuses on the solution of the combinatorial problems considered by using the D-WAVE solver.  Numerical simulation results are provided in Section \ref{sec5}. Our conclusions are discussed in Section \ref{sec6}.

\begin{figure}
\centering
\includegraphics[width=\columnwidth]{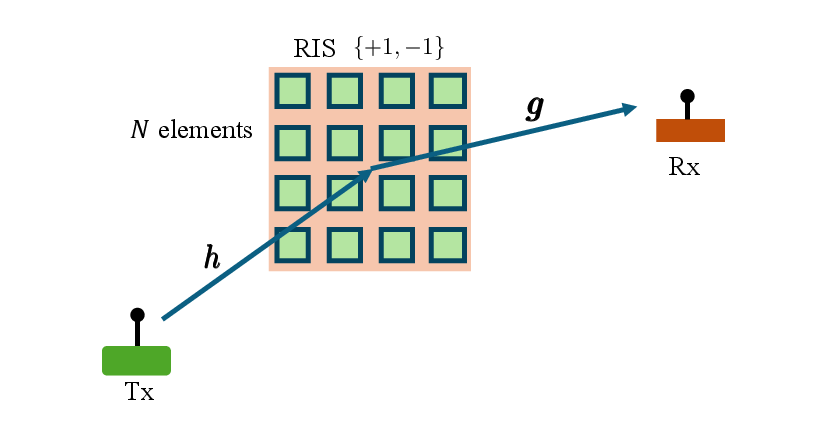}
\vspace{-1cm}
\caption{The system model consisting of one transmitter, one receiver, and a passive RIS with $N$ reflecting elements and $1$-bit resolution.}\label{sys_mod}
\end{figure}

\section{System model}\label{sec2}

We assume a simple RIS-aided communication setup consisting of a single-antenna transmitter, a single-antenna receiver, and a passive RIS with $N$ reflecting elements. Suppose that the channel between the transmitter and the RIS is denoted by $\pmb{h} \in \mathbb{C}^{N\times 1}$, while the channel between the RIS and the receiver is given by $\pmb{g} \in \mathbb{C}^{N\times 1}$ (standard cascaded RIS channel model). Due to obtacles and deep shadowing, a line of sight (LOS) link does not exist and communication between the transmitter and the receiver is performed through the RIS \cite{BAS,BAY}. We assume independent and identically distributed (i.i.d.) Rayleigh fading, so that the entries of $\pmb{h}$ and $\pmb{g}$ are independent circularly symmetric complex Gaussian random variables with zero mean and unit variance. To reduce power consumption and signalling, the phase shifts of the RIS have binary resolution and thus can take two discrete values {\it i.e.,} $+1$ and $-1$ corresponding to the phase shifts $0$ and $\pi$, respectively. If $P_t$ is the transmit power, the received SNR can be written as
\begin{align}
\Gamma&=\frac{P_t|\pmb{g}^H \diag(\pmb{x}) \pmb{h} |^2}{N_0}, \nonumber \\
&=\frac{P_t \pmb{h}^H \diag(\pmb{x}^T)\pmb{g}\pmb{g}^H\diag(\pmb{x})\pmb{h}}{N_0} \nonumber \\
&=\frac{P_t \pmb{x}^T\big(\diag(\pmb{h}^H)\pmb{g}\pmb{g}^H\diag(\pmb{h})  \big) \pmb{x}}{N_0} \nonumber \\
&=\frac{P_t \pmb{x}^T \pmb{R} \pmb{x}}{N_0},
\end{align}
where $\pmb{R}=\diag(\pmb{h}^H)\pmb{g}\pmb{g}^H\diag(\pmb{h})$,  $\pmb{x}=[x_1, x_2, \ldots, x_N]^T$ denotes the binary RIS phase shift vector, with each element $x_i \in \{+1,-1\}$ representing a discrete phase shift of $0$ or $\pi$, respectively, imposed by the $i$-th RIS element, and $N_0$ is the variance of the additive white Gaussian noise (AWGN) at the receiver. Global channel state information (CSI) is assumed at both the transmitter and the receiver, which can be acquired using techniques such as those proposed in \cite{GZ,FF}. Maximizing the received SNR directly increases the achievable Shannon capacity, since the capacity expression $C=\log_2(1+\textrm{SNR})$ is a strictly increasing function of SNR under Gaussian signaling. Fig. \ref{sys_mod} schematically depicts the system model.

\subsection{ Preliminaries: Ising-to-QUBO Transformation}

The optimization problems encountered in this work naturally arise in the Ising formulation, since the RIS phase shift variables are constrained to the spin set $\{+1,-1\}$. This representation closely reflects the physical nature of the system and aligns with the native spin-based variables used in quantum annealing hardware, which operates according to the Ising spin-glass model. The standard Ising form is given by
\begin{align}
\textrm{(Ising problem)}\;\;\;\;\min_{\pmb{x}} \sum_{i}u_i x_i+ \sum_{i<j} J_{i,j}x_ix_j, \label{isi1}
\end{align}
where the real coefficients $u_{i}$ and $J_{i,j}$ denote the linear (bias) and the quadratic (coupler strength) terms, respectively, $x_i$ are spin variables taking values in the set $\{+1,-1\}$. To reformulate the problem in a form suitable for our computational framework, we adopt the equivalent QUBO formulation by applying the transformation $x_i\rightarrow 2b_i-1$ \cite{JAM1} {\it i.e.,}
\begin{align}
\textrm{(QUBO problem)}&\;\;\min_{\pmb{b}}\sum_{i}G_{i,i}b_i+\sum_{i<j}G_{i,j}b_ib_j \nonumber \\
&=\;\min_{\pmb{b}}\; \pmb{b}^T \pmb{G}\pmb{b},
\end{align}
where $b_i \in \{0,1\}$ are binary variables, $G_{i,j}=4J_{i,j}$ and $G_{i,i}=2u_i-2\sum_{k<i}J_{k,i}-2\sum_{i<k}J_{i,k}$. QUBO/Ising problems belong to NP-hard complexity class. Due to the equivalence between NP-hard problems, many optimization problems can be represented as QUBO/Ising {\it e.g., traveling salesman, nurse scheduling, max-cut, etc.}. Although the Ising and QUBO formulations are mathematically equivalent, the QUBO representation is more commonly supported by commercial solvers, allows for easier embedding of constraints, and is more directly supported by the D-WAVE Ocean SDK and its associated tools (the quantum annealing platform used in our experimental results). For these practical reasons, we adopt the QUBO form in our study, while still referring to the Ising model where it offers clearer physical or theoretical insight ({\it e.g.,} in the context of adiabatic quantum evolution in Section \ref{dw}). To facilitate the reformulation process, we now present four key lemmas that convert Ising-type terms (both objective functions and constraints) into equivalent QUBO representations.

\begin{lemma}\label{lma}
The quadratic Ising term $\pmb{x}^T \pmb{J}\pmb{x}$, where $\pmb{x}$ is a spin vector with elements in $\{+1,-1\}$, and $\pmb{J}$ is a hermitian $N\times N$ matrix, is equivalent to the quadratic binary expression $\pmb{b}^T \pmb{J}' \pmb{b}+C$, where $\pmb{b}$ is a binary vector, $\pmb{J}'=\Re\left(4\pmb{J}-4\diag(\pmb{J}\pmb{1})\right)$ and $C$ is a real-valued constant.
\end{lemma}

\begin{proof}
The proof is given in Appendix \ref{apA}. 	
\end{proof}

\begin{lemma}\label{lmb}
The linear Ising term $\pmb{x}^T\pmb{a}$, where $\pmb{x}$ is a spin vector with elements in $\{+1,-1\}$, and $\pmb{a}$ is a real-valued vector, is equivalent to the binary quadratic expression $\pmb{b}^T \diag(-2\pmb{a})\pmb{b}+C$, where $\pmb{b}$ is a binary vector and $C$ is a real-valued constant.
\end{lemma}

\begin{proof}
The proof is given in Appendix \ref{apB}. 	
\end{proof}

\begin{lemma}\label{lmc}
The linear equality $\pmb{x}^T\pmb{a}=c$, where $\pmb{x}$ is a spin vector with elements in $\{+1,-1\}$, $\pmb{a}$ is a real-valued vector, and $c>0$ is a constant, is equivalent to the binary quadratic equality $\pmb{b}^T \pmb{J}\pmb{b}+C=0$, where $\pmb{b}$ is a binary vector, $\pmb{J}=4\pmb{a}\pmb{a}^T+4(c-\pmb{a}^T\pmb{1})\diag(\pmb{a})$ and $C=c(c-2\pmb{a}^T\pmb{1})$. 
\end{lemma}

\begin{proof}
The proof is given in Appendix \ref{apC}. 	
\end{proof}

\begin{lemma}\label{lmd}
The quadratic equality $\pmb{x}^T\pmb{J}\pmb{x}=c$, where $\pmb{x}$ is a spin vector with elements in $\{+1,-1\}$, $\pmb{J}$ is hermitian matrix and $c>0$ is a constant, is equivalent to the binary quadratic equality $\pmb{b}^T \pmb{J}'\pmb{b}+C$ where $\pmb{b}$ is a binary vector, $\pmb{J}'=4(\Re(\pmb{J})-\frac{c}{N}\pmb{I})-4\diag((\Re(\pmb{J})-\frac{c}{N}\pmb{I})\pmb{1})$ and $C$ is a real-valued constant.   
\end{lemma}
\begin{proof}
	The proof is given in Appendix \ref{apD}. 	
\end{proof}

\section{RIS-aided index modulation}\label{sec3}

In this section, we present the conventional passive beamforming solution with binary phase shifts and the proposed RIS-based IM scheme. The associated QUBO problems are introduced and theoretical expressions for the average received signal power and average capacity are provided.

\subsection{Conventional beamforming}
The conventional beamforming scheme with binary phase shifts, designs the phase shift vector to maximize the SNR at the receiver. It corresponds to the solution of the following quadratic (Ising) unconstrained optimization problem 
\begin{align}
\max_{\pmb{x} \in \{+1,-1 \}^N} \pmb{x}^T \pmb{R} \pmb{x}\;=\min_{\pmb{x} \in \{+1,-1 \}^N} \pmb{x}^T (-\pmb{R}) \pmb{x}. \label{p1}
\end{align}
By using the expression in Lemma \ref{lma}, the problem in \eqref{p1} can be written as QUBO {\it i.e.,} 
\begin{align}
(P1)\;\;	&\min_{\pmb{b} \in \{0,1 \}^N} \pmb{b}^T \Re\big( -4\pmb{R}+4\diag(\pmb{R}\pmb{1})\big) \pmb{b} \nonumber \\
&=	\min_{\pmb{b} \in \{0,1 \}^N} \pmb{b}^T \pmb{Q} \pmb{b}, \label{p2}
\end{align}
where $\pmb{Q}=\Re\big(-4\pmb{R}+4\diag(\pmb{R}\pmb{1}) \big)$ and the constant terms can be ignored from the optimization. 

\subsection{IM-based beamforming}

The proposed technique combines the conventional RIS-based passive beamforming with IM. Specifically, in addition to the conventional modulation which is employed at the transmitter, we use the number (cardinality) of $+1$s in the phase vector to convey extra bits of information by using the principles of IM. One could also consider the number of $-1$s in the phase vector, which is equivalent due to the symmetry in the SNR expression {\it i.e.,} $|\pmb{g}^H\diag(\pmb{x})\pmb{h}|^2=|\!-\pmb{g}^H\diag(\pmb{x})\pmb{h}|^2$. As such, the distinct beamforming vectors can have $k \in \{0,1,\ldots,\lfloor{N/2}\rfloor\}$ phase shifts equal to $+1$. Hence, the transmitter can adjust the index $k$ according to the input information, thus conveying $\log_2(\lfloor{N/2}\rfloor+1)$ extra bits at each transmission time. To further boost the performance, the RIS vector that maximizes the SNR at the receiver and has $k$ values equal to $+1$ is selected. It is worth noting that the IM squeezes the space dimension of the passive beamforming from $2^{N-1}$ to $\binom{N}{k}$ vectors at each transmission. Therefore, even though the performance of the passive beamforming is reduced, it is compensated by the IM performance gain. The proposed IM-based design corresponds to the following constrained binary optimization problem 
\begin{align}
&\min_{\pmb{x} \in \{+1,-1 \}^N} \pmb{x}^T (-\pmb{R}) \pmb{x} \\
&\textrm{subject to}\;\;\;\; \pmb{x}^T\pmb{1}=2k-N, \label{con}
\end{align}
where the objective function is the SNR at the receiver, the constraint in \eqref{con} refers to the IM that enforces $k$ RIS elements with phase shift $+1$ and thus $N-k$ elements with phase shift $-1$ {\it i.e.,} $\pmb{x}^T\pmb{1}=\sum_{i=1}^N {x_i}=(+1)\times k+(-1)\times(N-k)=2k-N$, where $k \in \{0,1,\ldots,\lfloor{N/2}\rfloor \}$.  By applying the expressions in Lemmas \eqref{lma} and \eqref{lmc}, the above optimization problem can be converted to an equivalent binary form, where both the objective and the constraint are quadratic binary functions {\it i.e.,}
\begin{align}
&(P2)\;\;\min_{\pmb{b}\in \{0, 1\}^N} \pmb{b}^T \pmb{Q} \pmb{b} \nonumber  \\
&\textrm{s.t.}\;\;\;\;\;\;\; \pmb{b}^T \pmb{T}\pmb{b}+C=0,\;\;\;(\textrm{equivalent to}\; \pmb{b}^T\pmb{1}+C'=0), \label{co1} 
\end{align} 
where $\pmb{T}=\big(4\pmb{1}\pmb{1}^T+8(k-N)\pmb{I} \big)$, and $C=(2k-N)(2k-3N)$. We note also that the constraint in \eqref{co1} can be written in a linear binary form as $\pmb{b}^T\pmb{1}+C'=0$, where $C'=k-N$.

To be able to detect the IM information, the receiver should have a mechanism to extract the RIS beamforming vector $\pmb{x}$ (and thus the parameter $k$) from the received signal. Since the receiver can measure the received SNR and given that a global CSI is assumed, IM detection corresponds to the solution of the following quadratic equation with respect to the spin vector $\pmb{x}$ 
\begin{align}
\frac{P_t}{N_0} \pmb{x}^T\pmb{R}\pmb{x}= \Gamma_m,
\end{align}
where $\Gamma_m$ denotes the measured SNR at the receiver; in this work, we assume perfect SNR estimation ($\Gamma_m = \Gamma$). It is also worth noting that distinct RIS phase vectors correspond to distinct SNR values, thus ensuring that the equation has a unique solution\footnote{Two distinct configurations $\pmb{x} \neq \pm \pmb{x}'$ yield the same SNR if and only if  $|\pmb{v}^H \pmb{x}|^2 = |\pmb{v}^H \pmb{x}'|^2$ with $\pmb{v} = \diag(\pmb{h}^H) \pmb{g}$. For a fixed pair $(\pmb{x},\pmb{x}')$, this is one real equation in the  $2N$ real parameters of $\pmb{v}$, so the solution set is a $(2N-1)$-dimensional surface 
inside the $2N$-dimensional space. With i.i.d.\ continuous fading 
(e.g., Rayleigh), the channel vector $\pmb{v}$ is drawn from a continuous distribution, 
so the probability of lying exactly on such a surface is zero. Thus, exact SNR collisions 
occur with probability zero (i.e., almost surely), although they can arise for particular 
deterministic channels.}. The above equation can be written as a feasibility problem {\it i.e.,}
\begin{align}
&\min_{\pmb{x}\in \{+1,-1\}^N} 0 \label{op2}\\
&\textrm{s.t.}\; \pmb{x}^T \pmb{R}\pmb{x}=\Gamma_m',
\end{align}
where $\Gamma_m'=\Gamma_m N_0/P_t$. By taking into account Lemma \ref{lmd}, the feasibility problem into \eqref{op2} can be written in an equivalent binary form
\begin{align}
&(P3)\;\min_{\pmb{b}\in \{0,1\}^N} 0 \\
&\textrm{s.t.}\;\; \pmb{b}^T \pmb{W} \pmb{b}=0,
\end{align}
where $\pmb{W}=\Re\left(4[ \pmb{J}-\frac{\Gamma_m'}{N}\pmb{I}]-4\diag([\pmb{J}-\frac{\Gamma_m'}{N}\pmb{I}]\pmb{1}) \right)$. 
The feasibility problem (P3) can be converted to a simple QUBO problem (by using the penalty method that will be discussed in Section \ref{penalty_method}) {\it i.e.,}
\begin{align}
&(P4)\;\min_{\pmb{b}\in \{0,1\}^N} \pmb{b}^T \pmb{W} \pmb{b}.
\end{align}
The above problem has similar structure to (P1). Therefore, although our experimental results concerns (P1), the methodology and key observations also hold for (P4). We emphasize that, in the proposed scheme, decoding the IM index does not require demodulating the conventional symbol. The receiver only estimates the received SNR, which (under global CSI) maps uniquely to the RIS vector and thus to $k$, enabling index recovery via a feasibility-type QUBO problem. As such, the IM layer operates independently of signal demodulation and remains decodable even when the signal component is weak or unreliable. This provides a form of robust auxiliary communication, particularly useful in low-SNR or energy-constrained scenarios.

\subsection{Capacity bounds}\label{bounds}
Here, we provide theoretical bounds on the average Shannon capacity achieved by the proposed technique and compare the performance with the conventional beamforming scheme. Assuming the transmitted IM information is equiprobable, the average capacity can be written as
\begin{align}
\C &= \frac{1}{\lfloor N/2\rfloor +1}\sum_{k=0}^{\lfloor N/2\rfloor}\mathbb{E}(\log_2(1+\Gamma^*_k)) + \log_2(\lfloor N/2\rfloor+1),\label{capacity}
\end{align}
where $\Gamma^*_k$ is the maximum SNR attained with $k \in \{0,1,\ldots,\lfloor{N/2}\rfloor\}$ phase shifts equal to $+1$, given by
\begin{align}
\Gamma^*_k = \frac{P_t}{N_0} {\pmb{x}^*}^T \pmb{R} \pmb{x}^* = \frac{P_t}{N_0} \left|\sum_{i=1}^N |h_i| |g_i| x^*_i e^{j(\theta_i+\phi_i)}\right|^2,
\end{align}
with $\pmb{x}^*$ being the vector that maximizes the channel gain with $\sum_{i=1}^N x_i = 2k-N$, $\theta_i$ and $\phi_i$ are the phases of $h_i$ and $g_i$, respectively. In what follows, we assume $h_i, g_i \sim \mathcal{CN}(0,1)$ and denote by $H_k = \mathbb{E}({\pmb{x}^*}^T \pmb{R} \pmb{x}^*)$ the average channel gain. Before proceeding with the analytical evaluation of $H_k$, we first look at the conventional beamforming case with binary phase shifts. We assume that the phase shift design assigns the discrete phase shift that is nearest to the optimal continuous value, {\it i.e.,} $-(\theta_i+\phi_i)$ \cite{ZHANG}. Therefore, based on this approach, it is known that the average channel gain for the conventional beamforming is \cite{ZHANG}
\begin{align}\label{conv_bf}
H_c = N+\frac{1}{4}N(N-1).
\end{align}
In contrast, the proposed IM-based beamforming technique does not allow for all elements to set their phases as close to the optimal continuous value, which instills a loss in performance. The following proposition provides the average channel gain for this case.

\begin{proposition}\label{prop}
The average channel gain achieved by the IM-based beamforming technique with $N$ elements and $k$ phase shifts equal to $+1$ is
\begin{align}\label{IM_bf}
H_k = N + \frac{1}{4}N(N-1) - S_k,
\end{align}
where, for odd $N$, we have
\begin{align}\label{s1}
S_k = \frac{1}{2^{N-1}} \sum_{n=0}^{\lfloor N/2 \rfloor} \binom{N}{n} |k-n| (N-|k-n|),
\end{align}
whereas, for even $N$, we have
\begin{align}\label{s2}
S_k &= \frac{1}{2^{N-1}} \sum_{n=0}^{N/2-1} \binom{N}{n} |k-n| (N-|k-n|)\nonumber\\
&\quad+\frac{1}{2^N} \binom{N}{N/2} \left|k-\frac{N}{2}\right|  \left(N-\left| k-\frac{N}{2}\right|\right).
\end{align}
\end{proposition}

\begin{proof}
The proof is given in Appendix \ref{prop_prf}.
\end{proof}

Comparing \eqref{conv_bf} with \eqref{IM_bf}, we can observe that $S_k$ characterizes the loss in channel gain due to the constraints imposed to the RIS vector by the IM-based beamforming technique. This loss in channel gain is inversely proportional to $k$. Indeed, for $k=0$ the loss is maximized with $S_0 = \frac{1}{4}N(N-1)$, resulting in $H_0 = N$, which corresponds to a random phase shift configuration \cite{CP}; this is expected, since only one vector is available for this case. On the other hand, values of $k$ close to $\lfloor N/2 \rfloor$ provide smaller values for $S_k$. It is important to note that although the channel gain is reduced, the IM-based technique compensates by enhancing the capacity (see \eqref{capacity}). To gain further insights, we look at the ratio between $H_k$ and $H_c$ in the asymptotic regime $N\to\infty$.

\begin{corollary}
As $N\to\infty$, we have
\begin{align}
\epsilon_k = \frac{H_k}{H_c} \to 1 - \frac{4}{N^2} S_k,
\end{align}
where $S_k$ is given by \eqref{s1} or \eqref{s2}.
\end{corollary}

From the above, we can conclude that $\epsilon_0 \to 0$, as expected. However, we have $\epsilon_{\lfloor N/2 \rfloor} \to 1$, which shows that, asymptotically, the IM-based scheme can match the performance of the conventional scheme. We can now provide a bound on the average Shannon capacity. Using \eqref{capacity}, we have
\begin{align}
&\C \leq \frac{1}{\lfloor N/2\rfloor +1} \sum_{k=0}^{\lfloor N/2\rfloor}\log_2(1+\mathbb{E}(\Gamma^*_k)) + \log_2(\lfloor N/2\rfloor+1)\nonumber\\
&=\frac{1}{\lfloor N/2\rfloor +1}\sum_{k=0}^{\lfloor N/2\rfloor}\log_2\left(1+\frac{P_t}{N_0}H_k\right) + \log_2(\lfloor N/2\rfloor+1),\label{cap}
\end{align}
which follows with the use of Jensen's inequality \cite[Ch. 2.6]{COV} and $H_k$ is given in Proposition \ref{prop}.

\subsection{A toy example}

We present a simple/toy numerical example to demonstrate the basic characteristics of the proposed scheme. We consider an RIS with $N=5$ elements, $P_t=1$, $N_0=1$ and constant wireless channels. Table \ref{table1} presents the RIS vectors for each index ($k\in \{0,1,2 \}\}$ since $\lfloor 5/2 \rfloor=2$) and the associated SNRs.

The conventional beamforming scheme selects the RIS vector that achieves the maximum SNR. Specifically, among the $2^{N}/2=2^{N-1}=2^{4}=16$ distinct vectors (due to the symmetry in the SNR expression) that are placed in the three columns of Table \ref{table1}, the conventional design will select the entry $\{+1,-1,-1,+1,-1\}$ with SNR equal to $1.584$; in this case the achieved Shannon capacity becomes $\mathcal{C}=\log_2(1+\textrm{SNR})=\log_2(1+1.584)=1.37$ bits per channel use (bpcu). 

On the other hand, the proposed IM scheme (at each transmission time), selects the RIS vector with the maximum SNR among the vectors of a specific column {\it e.g.,} for $k=1$, it selects the vector $\{-1,-1,-1,+1,-1\}$ with SNR $1.57$. If the IM data are equiprobable, the achieved average Shannon capacity can be written as $\mathcal{C}=\mathbb{E}(\log_2(1+\textrm{SNR}))+\log_2(3)=\frac{1}{3}\log_2(1+0.279)+\frac{1}{3}\log_2(1+1.57)+\frac{1}{3}\log_2(1+1.584)+\log_2(3)=2.6138$ bpcu. 

This toy example demonstrates the performance benefits of the proposed technique. We note that while the current framework assumes ideal SNR estimation at the receiver, in practical settings small differences in SNR values across candidate RIS configurations (as observed in Table \ref{table1}) may lead to index detection errors. Investigating the robustness of the IM scheme under SNR fluctuations and developing resilient detection strategies is an important direction for future work.

\begin{table}[t]
\caption{IM-based RIS vectors and associated SNR values; setup with $N=5$, $P_t=1$, $N_0=1$, wireless channels with  $\pmb{h}=[-0.048+0.0364j, -0.138+0.584j, -0.153+1.079j, -0.2143+0.3302j, 0.0163-0.1483j]$, and $\pmb{g}=[0.4421+0.0956j, 0.1296+0.3643j, -0.7282+0.1848j, 0.6712-0.6657j, 0.2171-0.1148j]$.}
\label{table1}
\centering
\resizebox{\columnwidth}{!}{
\begin{tabular}{|c|c|c|c|}
\hline
IM & $k=0$ & $k=1$ & $k=2$ \\ 
\hline
\hline
\multirow{9}{*}{\shortstack{RIS \\ vector \\ \& \\ SNR}} & \multirow{9}{*}{\shortstack{\bm{$\{-1,-1,-1,-1,-1\}$} \\ \bm{$0.279$}}}  & $\{-1,-1,-1,-1,+1\}$, $0.208$      &  {$\{+1,+1,-1,-1,-1\}$}, $0.324$   \\ 
\cline{3-4}
&     & \bm{$\{-1,-1,-1,+1,-1\}$}, \bm{$1.57$}      & $\{+1,-1,+1,-1,-1\}$,  $1.346$   \\ 
\cline{3-4}
&     & $\{-1,-1,+1,-1,-1\}$, $1.407$    & \bm{$\{+1,-1,-1,+1,-1\}$},  \bm{$1.584$}    \\ 
\cline{3-4}
&     & $\{-1,+1,-1,-1,-1\}$, $0.281$   & {$\{+1,-1,-1,-1,+1\}$},  $0.203$    \\ 
\cline{3-4}
&    & $\{+1,-1,-1,-1,-1\}$, $0.274$   & {$\{-1,+1,+1,-1,-1\}$},  $1.405$     \\ 
\cline{3-4}
&    &    & {$\{-1,+1,-1,+1,-1\}$},  $1.506$    \\ 
\cline{4-4}
&    &     & {$\{-1,+1,-1,-1,+1\}$},  $0.228$    \\ 
\cline{4-4}
&    &     & {$\{-1,-1,+1,+1,-1\}$},  $0.271$    \\ 
\cline{4-4}
&    &    & {$\{-1,-1,+1,-1,+1\}$}, $1.568$   \\ 
\cline{4-4}
&    &   & {$\{-1,-1,-1,+1,+1\}$},  $1.392$  \\ 
\hline
\end{tabular}}
\end{table}

\section{Optimal RIS design and D-WAVE heuristic}\label{sec4}

In this section, we deal with the solution of the optimization problems (P1)/(P4), and (P2). We note that (P1)/(P4) correspond to conventional QUBO problems, while (P2) is a constrained binary optimization problem where both the objective and the equality constraint are quadratic binary functions.

\subsection{QUBO/Ising problem and D-WAVE heuristic}\label{dw}

The problems (P1)/(P4) are conventional QUBO problems while the optimization problem in (P2) can be converted to appropriate QUBO form, as it will be presented in the following discussion. Therefore, the solution of a QUBO problem is vital for the proposed RIS-aided design schemes. In this work, we adopt QA to solve the QUBO problems considered and specifically the D-WAVE QA device that offers remote-based access (leap platform) to real-world D-WAVE quantum computers and appropriate source development tools (Python SKD Ocean) \cite{WAVE}.
 
QA is a specialized method of quantum computing aimed at solving QUBO/Ising problems by leveraging quantum mechanical principles (adiabatic theorem). Specifically, at the heart of QA process is the concept of \textit{adiabatic evolution}, a process where a quantum system remains in its ground state (the lowest quantum energy state) as it gradually evolves from an initial Hamiltonian $H_{\text{initial}}$ (representing an easy-to-solve problem) to a final Hamiltonian $H_{\text{final}}$ that represents the problem to be solved \cite{MCG,YAR}. If we consider the Ising problem in \eqref{isi1}, the QA process can be described as 
\begin{align}
H(t) = -\frac{A(t)}{2} H_{\text{initial}} + \frac{B(t)}{2} H_{\text{final}},
\end{align}
where $t \in [0, 1]$ denotes the anneal fraction, $A(t)$ (monotonic decreasing function) and $B(t)$ (monotonic increasing function) represent the anneal schedules (adiabatic path) which relay on the quantum processor \cite{ADI}, $H_{\text{initial}}=\sum_i \sigma_x^{(i)}$ , and $H_{\text{final}}=\sum_i u_i\sigma_z^{(i)}+\sum_{i<j}J_{i,j}\sigma_z^{(i)}\sigma_z^{(j)}$; $\sigma_x^{(i)}$ and $\sigma_z^{(i)}$ are $x$ and $z$ Pauli matrices acting on qubit $i$, respectively. 

D-WAVE (based in British Columbia, Canada) has developed a series of quantum processors designed to implement QA for real-world problems \cite{WAVE}; these processors are noisy versions of the ideal adiabatic evolution. The {\it Advantage\_System6.4} \cite{qpu} is the latest D-WAVE quantum processor that features over $5,500$ flux qubits, an increase from earlier quantum devices, making it one of the largest commercial QA available. These qubits are arranged in a Pegasus hardware topology (each qubit is connected to $15$ other qubits), a high-connectivity hardware graph structure designed to allow for more flexible and efficient embedding of optimization problems. The device has limited arithmetic resolution {\it i.e.,} the linear terms ($u_i$) and quadratic terms ($J_{i,j}$) are quantized in the interval $[-4,4]$ and $[-2,1]$, respectively.

The mapping of a physical QUBO problem (which represents the quadratic interconnection of binary variables) into the limited D-WAVE QA hardware topology/graph, is called minor embedding; it is an NP-hard problem and is mainly solved by using various heuristics. Since the hardware graph is not fully connected, minor embedding enables logical channeling between physical qubits to represent logical variables/qubits. The logical channeling is characterized by the strength of the logical links (called channel strength or ferromagnetic coupling) and it is a critical parameter for QA performance. In case a logical chain is broken at the end of QA process ({\it i.e.,} physical qubits that form a logical qubit have different final values), appropriate consensus algorithm is applied. 

Due to practical non-idealities ({\it e.g.,} hardware limitations, Hamiltonian noise, temperature fluctuations, etc.), the output of a single D-WAVE run (referred to as an anneal) is probabilistic and may be different than the ground state of $H_{\text{final}}$. To ensure an efficient solution for the considered QUBO problem, it is a common practice to solve the same QUBO instance multiple times; the best solution among all the anneals is the final D-WAVE QA output. The anneal time and the number of anneals are critical design parameters which are tuned empirically.

\begin{algorithm}[t]
\caption{Iterative penalty method to solve the problem in \eqref{opt2}. }\label{alg1}
\begin{algorithmic}[1]
\Require Initial $\mu$ and the increasing factor $\Delta \mu$, maximum number of iteration $i_{\text{max}}$; let $g(\pmb{b})=\pmb{b}^T\pmb{Q}\pmb{b}$.   
\Ensure Binary vector $\pmb{b}^*$ that solves the problem in \eqref{opt2}.
\For{$i \gets 1$ to $i_{\textrm{max}}$}
\State Solve \eqref{opt2} using D-WAVE  with $\pmb{b}\leftarrow\arg_{\pmb{b}} \min f(\pmb{b})$ and $f(\pmb{b})=\pmb{b}^T (\pmb{Q}+(\mu/2)\pmb{T})\pmb{b}$. 
\If{vector $\pmb{b}$ is feasible and $g(\pmb{b})<g(\pmb{b}^*)$ } \\ \;\;\;\;\;\;\;\;$\pmb{b}^*\leftarrow \pmb{b}$
\EndIf
\State Update $\mu\leftarrow \mu (\Delta\mu)$
\EndFor
\end{algorithmic}
\end{algorithm}

\subsection{Constrained optimization: Penalty method}\label{penalty_method}

The D-WAVE solver as well as all the existing commercial QUBO heuristics are not able to handle equality/inequality constraints directly. The most classical method to handle binary quadratic constraint is the penalty method, where the constraints are embedded in the objective function through penalty parameters \cite{YAR,KRI1}. Specifically, the constrained binary optimization problem in (P2), can be converted to the following QUBO form
\begin{align}
	&\min_{\pmb{b}\in \{0,1 \}}\pmb{b}^T\pmb{Q}\pmb{b}+\frac{\mu}{2} \pmb{b}^T \pmb{T}\pmb{b}, \nonumber \\
	&=\min_{\pmb{b}\in \{0,1 \}}\pmb{b}^T \left(\pmb{Q}+\frac{\mu}{2}\pmb{T} \right)\pmb{b}, \label{opt2}
\end{align}
where $\mu>0$ is the penalty parameter and the constant $C$ can be removed since it does not affect the solution. The above optimization problem is QUBO and can be solved using the D-WAVE QA device, as described in Section \ref{dw}. The parameter $\mu$ is critical for the performance of the penalty method; as $\mu$ increases, we enforce feasibility, but the quality of the solution becomes worse (and vice versa). It is also worth noting that large values $\mu$ are more sensitive to hardware limitations in coefficient resolution, which could result in further performance degradation. The optimal parameter $\mu$ can be tuned empirically through exhaustive numerical studies. To make this process more rigorous, we introduce Algorithm \ref{alg1}, where the parameter $\mu$ gradually increases while a QUBO problem is solved at each iteration. The best feasible solution over all iterations is the final solution of the penalty method.

\begin{algorithm}[t]
\caption{Iterative AL method to solve the problem in  \eqref{opt2}. }\label{alg2}
\begin{algorithmic}[1]
\Require \mbox{Initial} $\lambda,\mu$ \mbox{and the increasing factor} $\rho$.
\Ensure Binary vector $\pmb{b}^*$ that solves the problem in \eqref{opt2}.
\Repeat
\State Construct AL function $f(\pmb{b},\lambda,\mu)  =\pmb{b}^T(\pmb{Q}+\lambda\diag(\pmb{1})+(\mu/2) \pmb{T})\pmb{b}$.
\State Solve $\pmb{b} \leftarrow \arg\min_{\pmb{b}}f(\pmb{b},\lambda,\mu)$ using D-WAVE.
\If{($\pmb{b}^T\pmb{1}+C'>0$)}  update $\lambda\leftarrow \lambda+\mu(\pmb{b}^T\pmb{1}+C')$
\EndIf
\State Update $\mu = \rho\mu$.
\If{vector $\pmb{b}$ is feasible and $g(\pmb{b})<g(\pmb{b}^*)$ } \\ \;\;\;\;\;\;\;\;$\pmb{b}^*\leftarrow \pmb{b}$.
	\EndIf
\Until{stopping criteria are satisfied.}
\end{algorithmic}
\end{algorithm}

\subsection{Constrained optimization: Augmented Lagrangian method}   

The key weakness of the penalty method is that the optimal penalty parameter $\mu$ cannot be tuned automatically and requires exhaustive numerical studies. In addition, its value (if it is large) is very sensitive to the hardware arithmetic resolution/precision of the D-WAVE solver. On the other hand, the AL method combines Lagrangian optimization with the above penalty method in order to converge faster to the optimal solution while it avoids large values of the penalty parameter \cite{DJI}. In addition, it is associated with an iterative algorithm that allows to tune the auxiliary variables (Lagrangian multipliers and penalty parameters) automatically, through a rigorous iterative process. Based on the AL method, the optimization problem in (P2) can be written as 
\begin{align}
&\min_{\pmb{b}\in \{0,1 \}}\pmb{b}^T\pmb{Q}\pmb{b}+\lambda (\pmb{b}^T\pmb{1}+C')+\frac{\mu}{2} (\pmb{b}^T \pmb{T}\pmb{b}+C) \label{al} \\
&=\min_{\pmb{b}\in \{0,1 \}}\pmb{b}^T\pmb{Q}\pmb{b}+\lambda \pmb{b}^T\pmb{1}+\frac{\mu}{2}\pmb{b}^T \pmb{T}\pmb{b}, \nonumber  \\
&=\min_{\pmb{b}\in \{0,1 \}}\pmb{b}^T \left( \pmb{Q}+\lambda \diag(\pmb{1})+\frac{\mu}{2}\pmb{T}\right) \pmb{b}, \label{al1}
\end{align}   
where the second (linear) term in \eqref{al} refers to the Lagrangian multipliers method, the third (quadratic) term refers to the penalty method, and $\lambda$ is the Lagrangian multiplier; the constant terms can be ignored when solving the optimization problem. The above optimization problem can be solved iteratively, where a QUBO problem (given in \eqref{al1}) is solved at each iteration. The parameters $\lambda$ and $\mu$ are updated at each iteration with $\lambda \leftarrow \lambda+\mu (\pmb{b}^T\pmb{1}+C')$ and $\mu\leftarrow \rho \mu$ where $\rho>1$ is the increase factor (typically $\rho=1.1$ \cite{TAN}). Due to hardware limitations, the D-WAVE solver may occasionally return infeasible solutions. At each iteration, we consider the best feasible solution returned by the QUBO solver; if no feasible solution is found, we proceed with the best available (infeasible) one. The returned solution is used without post-processing, allowing us to directly assess the annealer's performance within the AL framework. Thanks to the iterative nature of the algorithm, occasional infeasibility does not significantly affect convergence. Exploring feasibility-enforcing variants or lightweight correction mechanisms is left for future work. The iterative AL algorithm is presented in Algorithm \ref{alg2}; the algorithm repeats until some formal termination criteria are met {\it e.g., } based on the number of iterations, time, convergence etc. 
Note that the local convergence of the AL method has been established when the objective and constraint functions are twice continuously differentiable, which is the case for the problem (P2), in \cite{AL-converge}.

\subsection{Performance benchmarks}

\subsubsection{Exhaustive search}

The exhaustive search (ES) can solve the optimization problems considered by computing the SNR objective function over all the possible RIS vectors; the vector that maximizes SNR (and/or satisfies the equality constraint in the case of (P2) is the solution of the optimization problem. Due to the symmetry of the SNR, the ES algorithm requires $2^{N}/2=2^{N-1}$ computations, and therefore its complexity is exponential with the number of RIS reflected elements. The ES can be applied for small number of RIS elements, while it becomes impractical for large configurations; it is used as a benchmark for the scenarios where it can be applied.  

\subsubsection{Simulated annealing}
Simulated annealing (SA) is an alternative classical method to solve QUBO problems. For our results, we use the the SA scheme of D-WAVE which is embedded in the Ocean SKD tool; SA operates locally on personal computers without necessitating data transmission to the D-WAVE hardware. For comparison, we employ SA in Step $2$ of Algorithm \ref{alg2} to implement the iterative AL algorithm and solve the problem in \eqref{opt2}. We set the parameters as $\lambda=2.1, \mu=2,\rho=1.1$, and the maximum number of iterations is $50$. This algorithm can be used as an approximate optimal solution when the above ES can not be applied (RIS configurations with a large number of elements).

\subsubsection{Random selection}

A trivial scheme that does not require any complicated computation is random selection. In this scheme, we generate $M$ random RIS vectors and select the best feasible vector (if it exists), that is, the vector that satisfies the IM parameter $k$ and returns the maximum SNR value. To have a fair comparison with the D-WAVE solver, $M$ is taken equal to the number of D-WAVE anneals.   

\subsection{Complexity}

Both the penalty (Algorithm \ref{alg1}) and the AL (Algorithm \ref{alg2}) algorithms are iterative and therefore their complexity is equal to the number of iterations multiplied with the complexity of D-WAVE/SA for a single iteration. Theoretically, within each iteration, the complexity of QA is $\mathcal{O}(e^{\sqrt{N}})$, while the complexity of SA is $\mathcal{O}(e^N)$ \cite{complexity}, where $N$ is the size of the QUBO problem. However, in practice, unlike the D-WAVE QA solver, SA is not affected by the quantum hardware topology and numerical resolution. It is also worth noting that in QA, instead of analyzing the computation time/complexity of a given algorithm, we mainly study the trade-off between the time and the probability that the QA output is correct. Therefore, in our D-WAVE QA experimental studies, the number of anneals as well as the duration of each anneal do not scale with the size of the problem and remain constant. Energy consumption is not explicitly analyzed in this work, as it depends on hardware-specific factors such as qubit count, annealing schedule, and embedding strategies, which are not accessible through the current D-WAVE interface. However, future D-WAVE systems are projected to offer notable power efficiency improvements over classical hardware in various wireless communication scenarios \cite{KASI}.

\section{Numerical results}\label{sec5}

Simulation and experimental results are carried-out to evaluate the performance of the proposed RIS-aided IM scheme. 

For the experimental results, we focus on four case studies that correspond to specific system parameters. Specifically, we consider four indicative RIS scenarios with $N=\{10, 25, 50, 100\}$. Note that the elements of the channel vectors $\pmb{h}, \pmb{g}$ are samples from a complex Gaussian distribution with zero mean and variance one. Fig. \ref{chann_mat} visualizes the (complex) matrices $\pmb{R}$ for the four scenarios considered. As for the IM parameter $k$, we consider $k=\{3,10,20,20\}$ for $N=\{10,25,50,100\}$, respectively, without loss of generality. The selected case studies are sufficient to demonstrate the effectiveness of the D-Wave quantum annealer in addressing the RIS design problem under consideration. Specifically, our experimental setup includes four representative, fully dense QUBO instances comprising $10$, $25$, $50$, and $100$ binary variables. These instances allow us to systematically assess the performance and scalability of the D-Wave system as the problem size increases. We note that our approach is general and can, in principle, be extended to larger RIS configurations. However, due to current quantum hardware limitations (including restricted access to the D-WAVE solver), we limit our experiments to RIS sizes up to $100$ elements.

\begin{figure}[t]
\includegraphics[width=1\columnwidth]{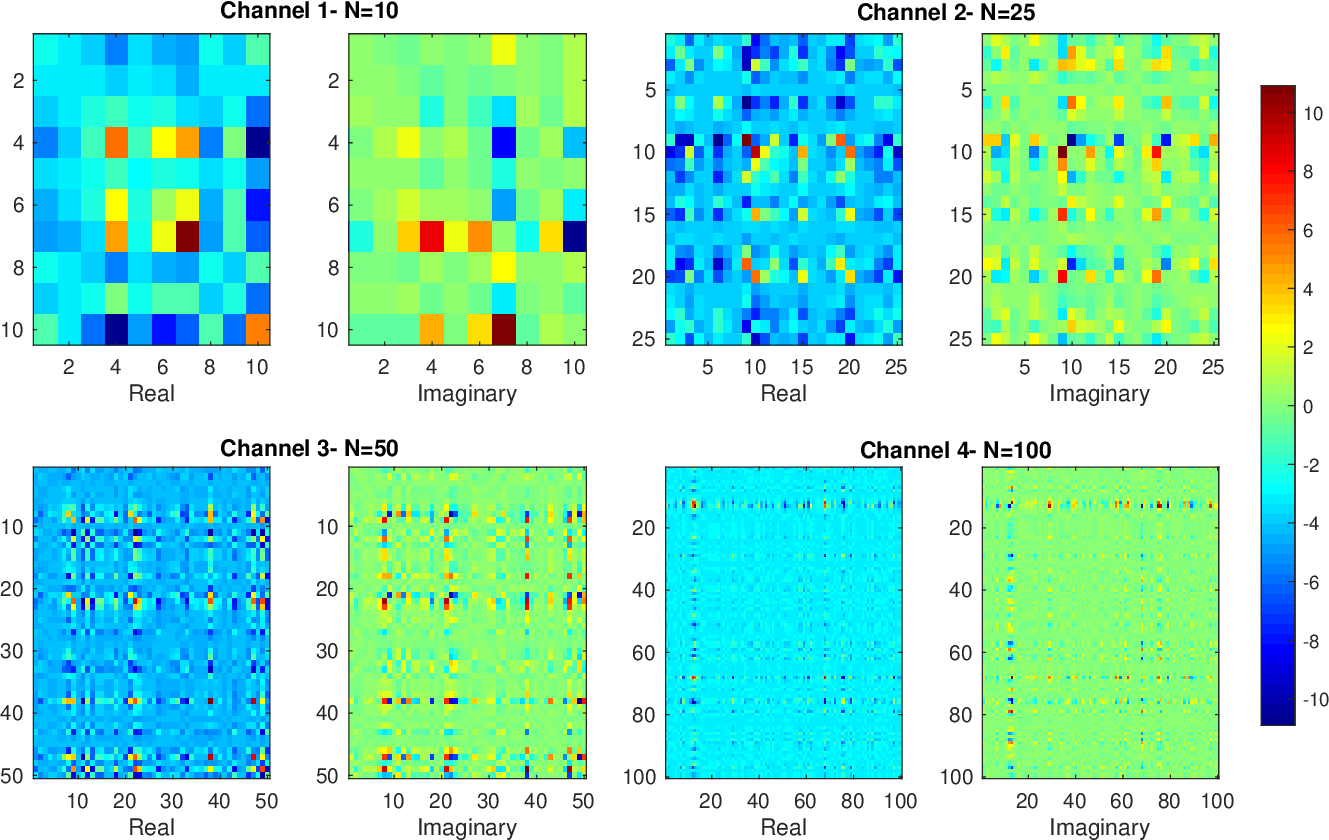}
\vspace{-0.4cm}
\caption{Matrix $\pmb{R}$ for the four considered cases studies with $N=\{10, 25, 50, 100\}$.}\label{chann_mat}
\end{figure}

\subsection{Proposed IM-based RIS design vs conventional RIS scheme}

In our first numerical results, we focus on the comparison between the proposed IM-based scheme and the conventional RIS design. Table \ref{table2} shows the average capacity performance (over $10,000$ channel realizations) for small RIS configurations where ES can be applied to solve (P1) and (P2). The key observation is that the proposed IM-based scheme significantly outperforms the conventional RIS design and the performance gain increases as the number of elements increases.  

For larger RIS configurations, we apply the SA algorithm with AL (Algorithm \ref{alg2}) for RIS vector optimization for both schemes; we consider the channel coefficients given in Fig. \ref{chann_mat}. Specifically, Fig. \ref{av1}(left) shows the normalized SNR ($\triangleq \pmb{x}^T\pmb{R}\pmb{x}$) performance for all values of the IM parameter ($k=0,\ldots, \lfloor N/2\rfloor$) as well as the SNR associated with the conventional RIS design (dashed line). An interesting remark is that the SNR performance of the IM-based scheme is mainly improved as $k$ increases. The justification is that the function $\binom{N}{k}$ increases monotonically (more combinations for larger $k$) and therefore the optimal RIS vector that maximizes the SNR is more probable to correspond to higher $k$. The results (on the right of Fig. \ref{av1}) show the capacity performance for a normalized set-up (with $P_t=1$ and $N_0=1$) for the four RIS configurations considered. The results are inline with our previous observations in Table \ref{table2}. The proposed IM-based scheme outperforms the conventional design and the performance gain increases as $N$ increases; the capacity gain increases from $2$ to $5$ bits per channel use for an RIS with $25$ and $100$ elements, respectively.

\begin{table}[t]
\caption{Average capacity performance (in bpcu) \\for IM-based and conventional RIS schemes;\\ RIS setup with $N=\{4, 6, 10, 14 \}$, $P_t=1$, $N_0=1$ and ES.}
\label{table2}
\centering
\resizebox{\columnwidth}{!}{
\begin{tabular}{|c||c|c|c|c|}
\hline
RIS scheme & $N=4$ & $N=6$ & $N=10$ & $N=14$ \\
\hline
\hline
IM-based  & $7.034$ & $8.558$ & $10.507$ & $11.821$  \\
\hline
Conventional  & $5.988$ & $7.086$ & $8.482$ & $9.408$   \\
\hline
\end{tabular}}
\end{table}

\begin{figure}[t]
\includegraphics[width=1\columnwidth]{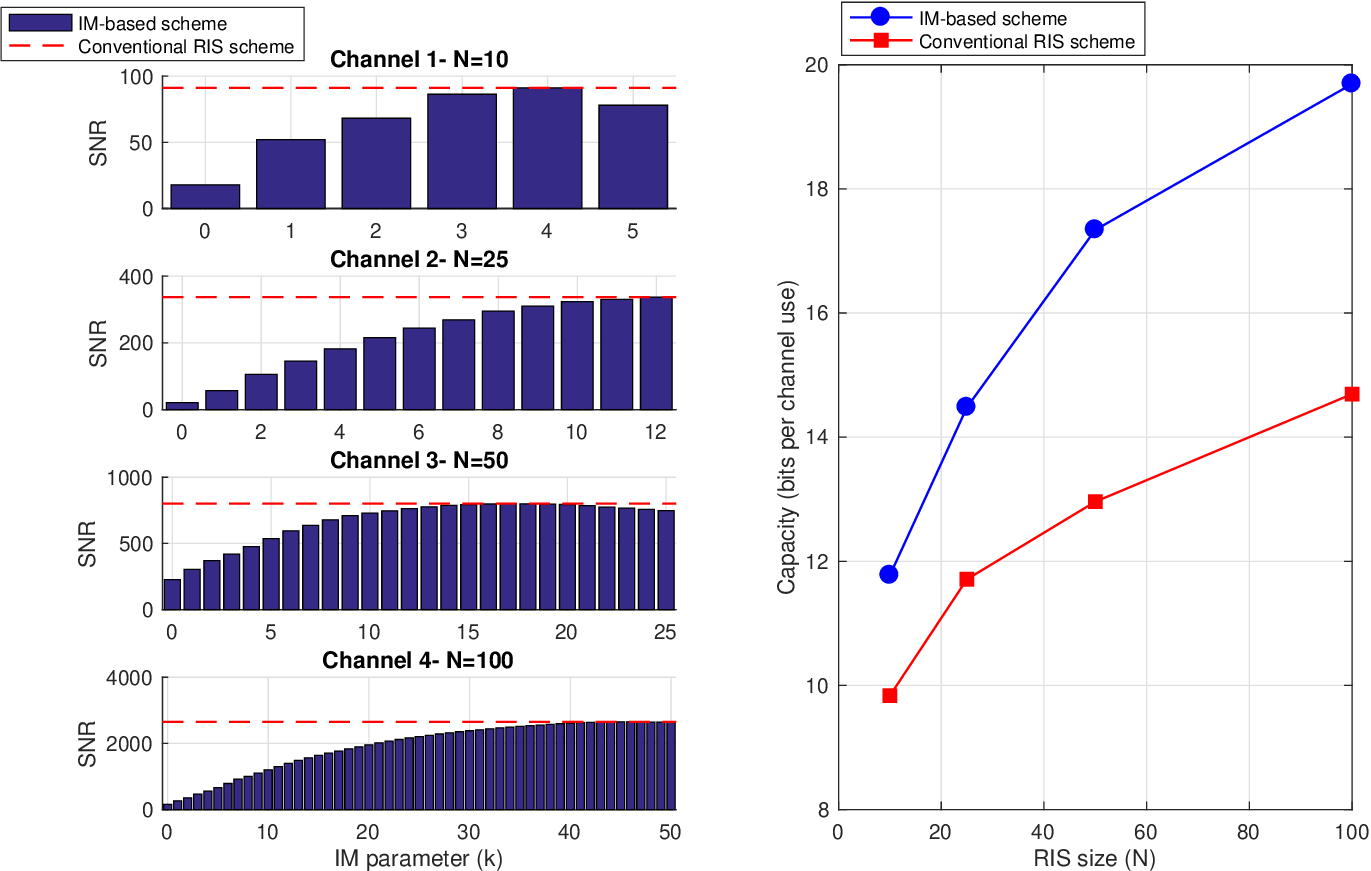}
\vspace{-0.4cm}
\caption{(left) Normalized SNR performance of the proposed IM-based RIS schreme versus the IM parameter $k$ for the four indicative RIS scenarios; the ``dashed line'' represents the conventional RIS scheme. (right) Capacity performance for both the IM-based and conventional RIS schemes for the four indicative RIS scenarios; $P_t=1$ and $N_0=1$.}\label{av1}
\end{figure}

\begin{table}[t]
\caption{Main parameters for the D-WAVE experiments}
\label{table3}
\centering
\begin{tabular}{|l|c|}
\hline
Number of anneals& $1,000$\\\hline
Anneal time & $1\mu$sec\\\hline
\makecell[l]{Channel strength\\ (ferromagnetic coupling parameter)}& $3$ \cite{KRI1}\\\hline
\end{tabular}\\\vspace{2mm}

\textbf{Penalty method}\\[1mm]
\begin{tabular}{|c|c|}
\hline
$\Delta\mu$ & $1.5$\\\hline
Maximum number of iterations & $20$\\
\hline
\end{tabular}\\\vspace{2mm}

\textbf{AL method}\\[1mm]
\begin{tabular}{|c|c|}
\hline
$\mu$ & $2$\\\hline
$\lambda$ & $2.1$\\\hline
$\rho$ & $1.1$\\\hline
Maximum number of iterations & $20$\\
\hline
\end{tabular}
\end{table}

\subsection{D-WAVE experimental results}
For our D-WAVE QA experiments, we use the D-WAVE Leap interface with the {\it Advantage System 6.4} quantum processing unit \cite{WAVE}. For each QUBO instance, we use the parameters listed in Table \ref{table3}. For minor embedding, we adopt the heuristic algorithm {\it minorminer} which is included in the Ocean SDK by default; in this case, a majority vote is applied to broken chains. Due to the probabilistic nature of the QA process, the D-WAVE solver may return several distinct solutions; the returned solutions are ordered in descending order of their objective function. In addition, we consider the probability of occurrence of each distinct solution (over the total number of anneals) and the probability of feasibility, which is the probability to obtain a feasible solution over all the anneals. All experimental results assume $P_t$ and $N_0 = 1$.

\begin{figure}[t]
\includegraphics[width=\columnwidth]{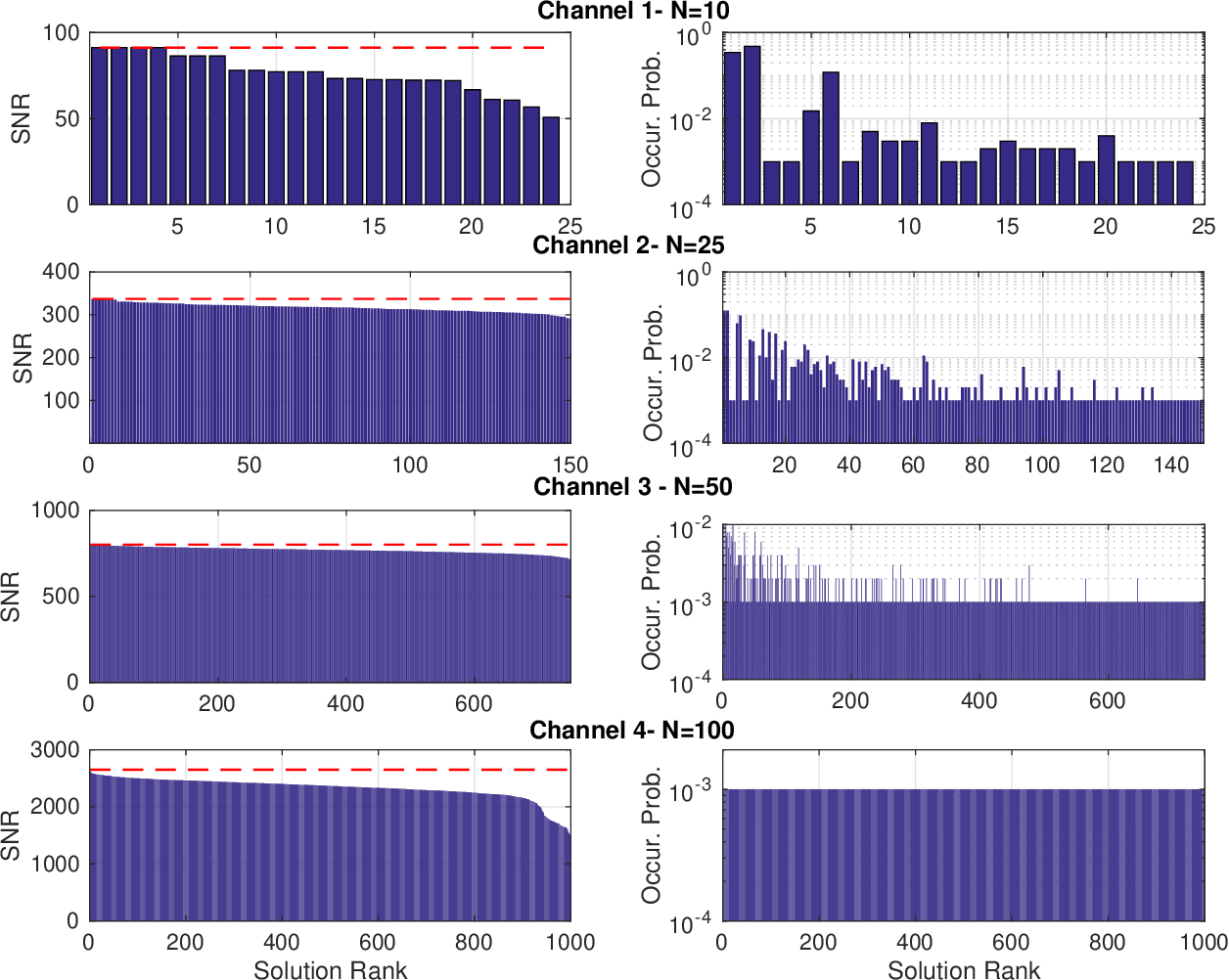}
\vspace{-0.4cm}
\caption{D-WAVE performance for the conventional RIS design with different numbers of RIS elements. (left) SNR for the ordered distinct solutions; "dashed line" represents the ES performance (for $N=10$) or SA performance (for $N>10$). (right) Probability of occurrence for the ordered distinct solutions.}\label{fig1}
\end{figure}

\begin{figure}[t]
\includegraphics[width=\columnwidth]{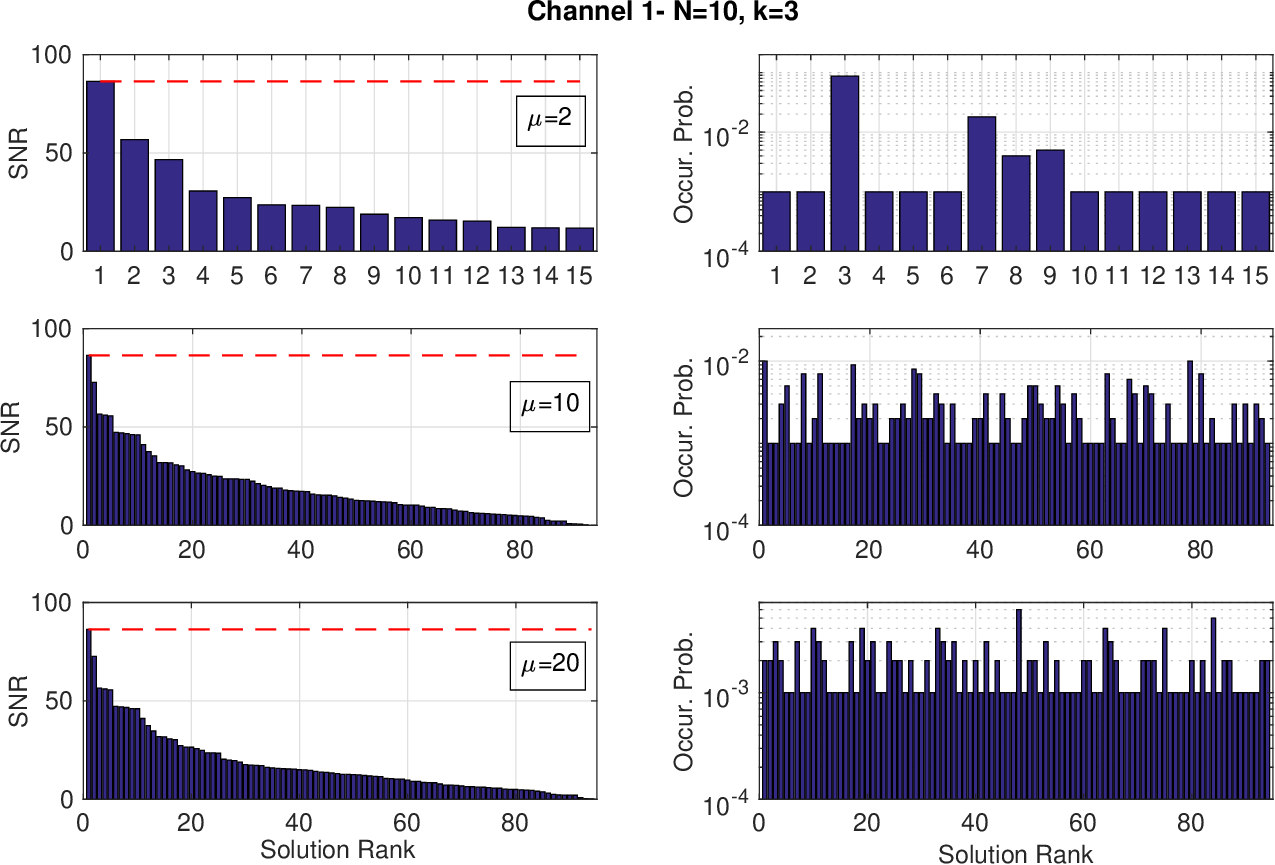}
\vspace{-0.4cm}
\caption{D-WAVE performance for the IM-based RIS design with the penalty method; setup with $N=10$, $k=3$ and $\mu=\{2, 10, 20 \}$. (left) SNR for the ordered feasible distinct solutions; "dashed line" represents the ES performance. (right) Probability of occurrence for the ordered feasible distinct solutions.}\label{fig2}
\end{figure}

\subsubsection{Conventional RIS design}
Fig. \ref{fig1} shows the performance of the D-WAVE solver for the four channels considered, when a conventional RIS optimization is applied. The first key observation is that the D-WAVE solver returns a close-to-optimal solution for all cases, which makes D-WAVE heuristic highly efficient for the combinatorial problem considered. A comparison between the four channels shows that the quality of the solutions reduces with the number of RIS elements. Specifically, we observe that the number of distinct solutions increases (for example, we have $24$ solutions for $N=10$, and $755$ solutions for $N=50$), and the probability of the best solution decreases as $N$ increases. It is worth noting that the mass distribution function of the returned solutions becomes symmetric for the scenario with $N=100$. 

\subsubsection{RIS-aided IM design - Penalty method}
Fig. \ref{fig2} deals with the solution of (P2) by using the penalty method for the simple scenario with $N=10$ and $k=3$. We plot the SNR performance of the feasible solutions ordered by D-WAVE and the associated probability of occurrence for three indicative penalty values $\mu=\{2, 10, 20\}$. We can see that the D-WAVE returns a close-to-optimal solution (similar to ES) for all penalty values. A more careful observation shows that as the penalty parameter increases, D-WAVE returns more feasible solutions (more notably from $\mu = 2$ to $\mu = 10$); this has been expected since a higher $\mu$ enforces feasibility. For $\mu=10$, D-WAVE returns almost $85$ feasible solutions, while the best solution has the highest probability of occurrence $\approx 0.01$. In addition, we observe that when $\mu=20$ although it further enforces feasibility, it decreases the quality of the solutions (the constraint penalty begins to dominate the objective, reducing the relative impact of the (negated) SNR term in the QUBO minimization.), {\it i.e.,} the probability of the occurrence of the best solution reduces to $0.002$.    

Fig. \ref{fig3} focuses on the more complex case studies with ($N=25$, $k=10$), ($N=50$, $k=20$), ($N=100$, $k=20$) when the iterative Algorithm \ref{alg1} is adopted. For the iterative algorithm, we use an initial penalty value $\mu=2$ and a maximum number of iterations equal to $20$; the penalty parameter is updated according to $\mu\leftarrow \mu (\Delta\mu)$ with $\Delta\mu=1.5$. The subfigures show the SNR of the best feasible solution, the average SNR for all feasible solutions, the probability of occurrence of the best solution, as well as the total probability of feasibility at each iteration. The first remark is that the D-WAVE solver with the penalty method can compute the optimal solution for $N=25$; $2\sim 3$ iterations seem to be sufficient to achieve a near-optimal performance and to maximize the probability of the occurrence of the best feasible solution. On the other hand, we observe that as $N$ increases, the quality of the solutions is significantly reduced. For $N=50$, D-WAVE returns a good solution in the first iteration, but its performance rapidly vanishes with the number of iterations. The main reason is that due to the limited hardware resolution of the D-WAVE solver, higher values of $\mu$ instead of enhancing the feasibility, reduce the quality of the returned solution and the feasibility. This issue becomes more critical for the scenario with $N=100$; in this case, the D-WAVE solver returns feasible solutions only for six penalty values while their quality is quite poor. Random selection provides a poor performance while it has feasibility issues as $N$ increases. Overall, the key message of this figure is that D-WAVE and penalty method seems to be efficient for low/intermediate RIS configurations while for high values of $N$, it becomes less efficient due to hardware limitations.

\begin{figure}[t]
\includegraphics[width=\columnwidth]{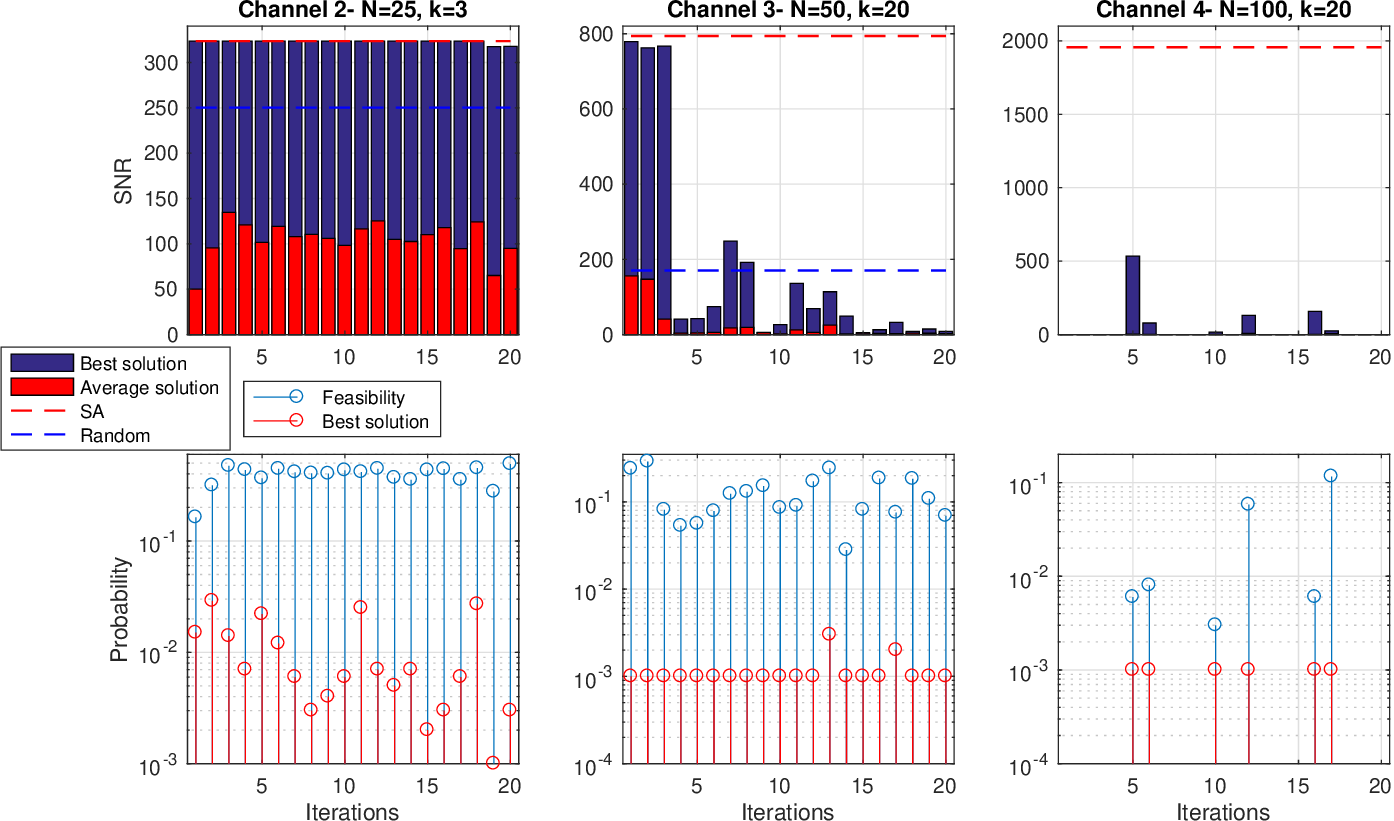}
\vspace{-0.4cm}
\caption{D-WAVE performance for the IM-based RIS design with the iterative penalty method; setup with ($N=25$, $k=3$), ($N=50$, $k=20$), ($N=100$, $k=20$), initial penalty $\mu=2$, $\Delta \mu=1.5$ and $20$ iterations. (top) SNR for the best feasible solution at each iteration- no results means that no feasible solution found; "(red) dashed line" represents the SA benchmark, "(blue) dashed line" represents the random selection benchmark. (bottom) Probability of occurrence for the best feasible solution and feasibility probability at each iteration.}\label{fig3}
\end{figure}

\subsubsection{RIS-aided IM design - AL method}
Fig. \ref{fig4} deals with the proposed iterative AL method for the solution of the problem (P2). For the AL, we consider the initial values $\mu=2$, $\lambda=2.1$, $\rho=1.1$, minimum number of iterations $5$, maximum number of iterations $20$, while the algorithm terminates when the change rate for three consecutive solutions is lower than $0.001$ or a maximum number of iterations is achieved. We consider the four indicative channel models and plot the performance of the D-WAVE solver in terms of SNR performance and occurrence probability for the best feasible solution (if exists) at each iteration. The first remark is that the D-WAVE solver with AL returns a nearly optimal solution for $N=\{10,25,50\}$ and a better solution than the penalty method for $N=100$. In addition, the average SNR performance of the returned solutions and the feasibility probability are improved compared to the penalty method as $N$ increases. The justification of this observation is that the AL method searches the solution space in a more flexible way (by using two free variables {\it i.e.,} $\lambda$ and $\mu$), while the free variables are increased in a more smooth way, which is more friendly to the D-WAVE limited hardware resolution than the penalty method. Although the quality of the solution decreases as the RIS elements increase (similar to the penalty method), the AL method is more efficient in integrating binary constraints into the D-WAVE QA device.

\begin{figure}
\includegraphics[width=\columnwidth]{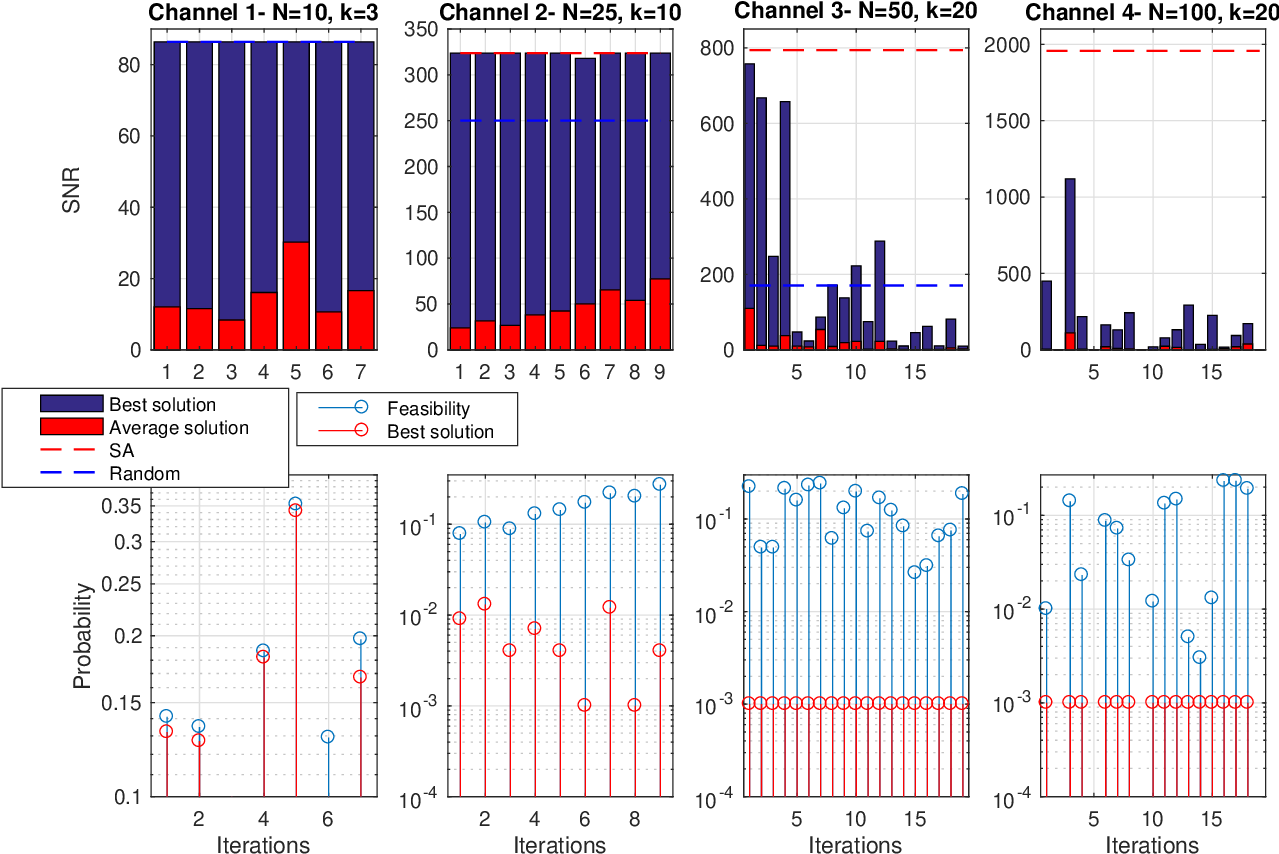}
\vspace{-0.4cm}
\caption{D-WAVE performance for the IM-based RIS design with the iterative AL method; setup with ($N=10$, $k=3$), (channel $N=25$, $k=3$), ($N=50$, $k=20$), ($N=100$, $k=20$), initial penalty $\mu=2$, initial $\lambda=2.1$, $\rho=1.1$, minimum iterations $5$, maximum iterations $20$; AL terminates when the change rate becomes smaller than $0.001$. (top) SNR for the best feasible solution at each iteration- no results means that no feasible solution found; "(red) dashed line" represents the SA benchmark, "(blue) dashed line" represents the random selection benchmark . (bottom) Probability of occurrence for the best feasible solution and feasibility probability at each iteration.}\label{fig4}
\end{figure}

\subsubsection{D-WAVE scalability and hardware limitations}
Our experimental results indicate that D-WAVE QA solution quality degrades with increasing the problem size, mainly due to hardware constraints such as limited qubits, sparse connectivity, finite precision, and environmental noise. Large and dense QUBO instances often exceed the native hardware connectivity, requiring minor embedding to map logical variables to chains of physical qubits. As these chains become longer, the system becomes more susceptible to noise and decoherence effects, often leading to distorted encoding and suboptimal solutions.

Despite these limitations, our results demonstrate that the D-WAVE solver achieves near-optimal performance for RIS configurations with up to $50$ elements, which lie within a practically relevant range given current hardware constraints and channel estimation complexity. Notably, for the conventional RIS scheme without IM, optimal solutions were achieved even with $100$ elements (see Fig.~\ref{fig1}). To further boost D-WAVE performance and scalability, techniques such as reverse annealing \cite{RQA} ({\it i.e.,} a warm-start strategy that begins from a good classical solution) can improve convergence and solution quality. Additionally, customized minor embedding methods may better utilize hardware topology and reduce resource overhead. These approaches offer promising directions for future research as D-WAVE QA hardware continues to advance.

\subsection{Capacity bounds}
We now provide numerical results based on our theoretical investigation in Section \ref{bounds} and compare the proposed IM-based scheme with the conventional beamforming scheme using binary phase shifts. For the IM-based scheme, the average SNR and capacity bound are evaluated using $P_t H_k/N_0$ and \eqref{cap}, respectively, where $H_k$ is given in Proposition \ref{prop}. For the conventional scheme, we consider $P_t H_c/N_0$ for the average SNR and $\log_2(1+P_t H_c/N_0)$ for the capacity bound, where $H_c$ is given by \eqref{conv_bf}. The analytical results are validated through Monte Carlo simulations with $10^6$ independent channel realizations.

In Fig. \ref{avg_snr}, we plot the average SNR with respect to the number of elements and different values of $k$, where simulation results are represented by lines and analytical results by markers. As discussed in Section \ref{bounds}, the case $k=0$ yields the lowest performance as its design corresponds to random phase shifts. On the other hand, values of $k$ close to $\lfloor N/2 \rfloor$ approach the performance of the conventional beamforming scheme, with the performance gap between the two decreasing as $N$ increases. Fig. \ref{avg_capacity} illustrates the average achieved Shannon capacity for $N=10$ and $N=25$. It can be clearly seen that the proposed IM-based scheme outperforms the conventional scheme, with significant gains observed at the low transmit power regime.

\begin{figure}
\includegraphics[width=0.9\columnwidth]{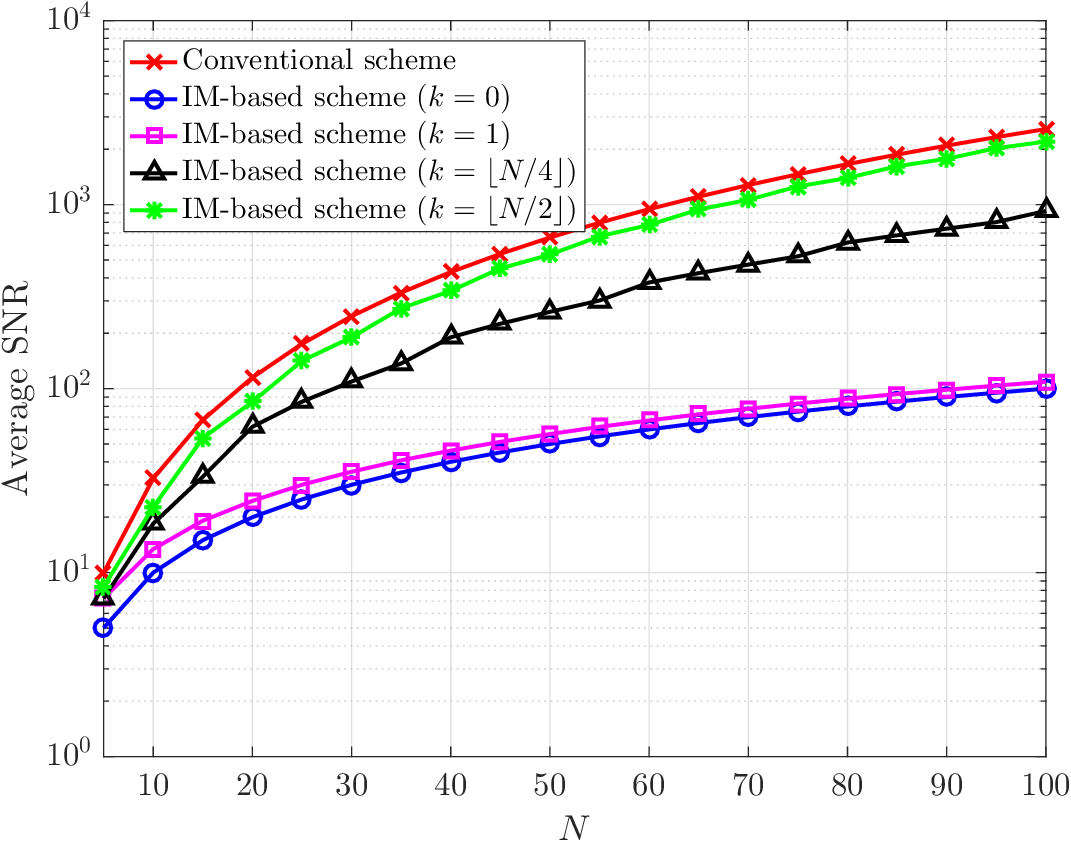}
\vspace{-0.1cm}
\caption{Average SNR with respect to $N$; $P_t = 1$ and $N_0 = 1$.}\label{avg_snr}
\end{figure}

\begin{figure}
\includegraphics[width=0.9\columnwidth]{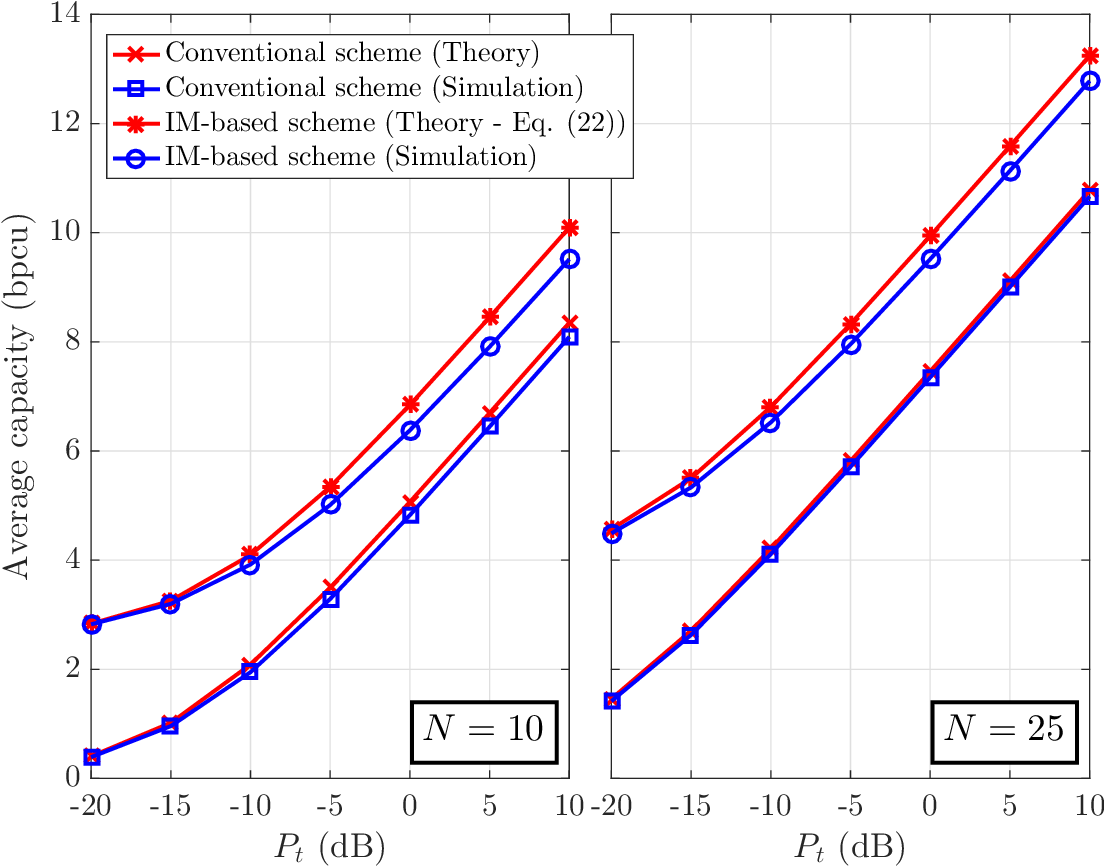}
\vspace{-0.1cm}
\caption{Average capacity with respect to $P_t$; $N_0 = 1$.}\label{avg_capacity}
\end{figure}

\section{Conclusion}\label{sec6}

This paper proposes a novel RIS-aided IM scheme for RIS configurations with $1$-bit resolution. In this approach, IM information bits are embedded within the RIS phase vector by indexing the binary phase shifts. To enhance performance, the RIS vector that maximizes the SNR at the receiver is selected, based on a predefined number of positive and negative phase shifts. The design is formulated as a combinatorial binary problem with an equality constraint at the transmitter, and we derive a theoretical upper bound for its average capacity. To address the constraint, we apply the penalty method alongside an iterative AL algorithm, which solves a QUBO problem at each iteration. We validate the IM technique and its mathematical framework using a real-world QA device from D-WAVE. Experimental results demonstrate that the proposed scheme outperforms conventional passive beamforming, and the gain becomes more significant as the number of RIS elements increases. The D-WAVE QA device efficiently solves small/intermediate RIS configurations, while the AL algorithm handles the equality constraint more effectively than the penalty method, especially given the device's limited numerical resolution. The proposed IM-based design can be extended to more complex networks, such as multi-user and multi-cell scenarios, by addressing user scheduling, interference management, and joint RIS optimization. Machine learning techniques (such as reinforcement learning or neural networks) could be employed to predict optimal RIS configurations and accelerate QUBO problem-solving via reverse annealing (warm-starting), thereby enhancing scalability and adaptability in dynamic environments. Another interesting direction is to explore RIS-only signaling schemes, where the RIS phase vector itself carries information, potentially enabling higher bit rates at the cost of increased decoding complexity and sensitivity to SNR variations.

\appendices

\section{Proof of Lemma \ref{lma}}\label{apA}

The quadratic term $\pmb{x}^T\pmb{J}\pmb{x}$ where $\pmb{J}$ is a hermitian matrix, can be converted to a binary quadratic term through a simple linear transformation. Specifically, by introducing $\pmb{x}=\pmb{1}-2\pmb{b}$ where $\pmb{b}$ is a binary vector with elements in $\{0,1 \}$, and $\pmb{1}$ is a full-ones vector, we have
\begin{align}
	\pmb{x}^T\pmb{J}\pmb{x}&=(\pmb{1}^T-2\pmb{b}^T)\pmb{J}(\pmb{1}-2\pmb{b}) \nonumber \\
	&=\pmb{1}^T\pmb{J}\pmb{1}-2 \pmb{b}^T\pmb{J}\pmb{1}-2\pmb{1}^T\pmb{J}\pmb{b}+4\pmb{b}^T \pmb{J} \pmb{b} \nonumber \\
&=4\pmb{b}^T\Re(\pmb{J})\pmb{b}-4\Re(\pmb{b}^T\pmb{J}\pmb{1})+\pmb{1}^T\pmb{J}\pmb{1} \label{in4}	\\
&=\pmb{b}^T\Re\big(4\pmb{J}-4\diag(\pmb{J}\pmb{1}) \big) \pmb{b}+\sum_{i,j}J_{i,j} \nonumber \\
&=\pmb{b}^T\pmb{J}'\pmb{b}+C,
\end{align}
where $\pmb{J}'$ is a symmetric (real) matrix and $C$ is a constant; the first term in the expression in \eqref{in4} is due to the following 
\begin{align}
\pmb{b}^T\pmb{J}\pmb{b}&=b_1\sum_{i}b_i J_{i,1}+b_2 \sum_{i}b_i J_{i,2}+\cdots+ \sum_{i}b_i J_{i,N} \nonumber \\
&=\sum_i b_i J_{i,i}+2\sum_{i,j}b_i b_j \Re(J_{i,j})\; =\pmb{b}^T\Re(\pmb{J})\pmb{b},
\end{align}
since $b_i^2=b_i$ (binary), $J_{i,j}=\overline{J_{j,i}}$, and $J_{i,i}$ are real (since $\pmb{J}$ is a hermitian matrix).

\section{Proof of Lemma \ref{lmb}}\label{apB}

By using the linear transformation between spin and binary variables (see Appendix \ref{lma}), the linear spin term $\pmb{x}^T\pmb{a}$ can be written as 
\begin{align}
	\pmb{x}^T\pmb{a}&=(\pmb{1}^T-2\pmb{b}^T)\pmb{a}\;=\pmb{1}^T\pmb{a}-2\pmb{b}^T\diag(\pmb{a})\pmb{b} \nonumber \\
	&=\pmb{b}^T\diag(-2\pmb{a})\pmb{b}+ \pmb{1}^T\pmb{a}.
\end{align}

\section{Proof of Lemma \ref{lmc}}\label{apC}
The linear spin equality $\pmb{x}^T\pmb{a}=c$ can be written as
\begin{align}
&\pmb{x}^T\pmb{a}=c \nonumber \\
&\Rightarrow (\pmb{x}^T\pmb{a}-c)^2=0 \nonumber \\
&\Rightarrow \pmb{x}^T\pmb{a}\pmb{a}^T\pmb{x}-2c\pmb{x}^T\pmb{a}+c^2=0  \nonumber \\
&\Rightarrow \pmb{b}^T \big(4\pmb{a}\pmb{a}^T-4\diag(\pmb{a}\pmb{a}^T\pmb{1}) \big)\pmb{b}+\pmb{1}^T\pmb{a}\pmb{a}^T\pmb{1} \nonumber \\
&\;\;\;\;-2c\pmb{b}^T\diag(-2\pmb{a})\pmb{b}-2c\pmb{1}^T\pmb{a}+c^2=0 \Rightarrow \label{inter1}\\
&\pmb{b}^T\bigg(4\pmb{a}\pmb{a}^T+4(c-\pmb{a}^T\pmb{1})\diag(\pmb{a}) \bigg)\pmb{b}+(c-\pmb{a}^T\pmb{1})^2=0, 
\end{align}
where the expression in \eqref{inter1} applies Lemmas \ref{lma} and \ref{lmb}. 

\section{Proof of Lemma \ref{lmd}}\label{apD}

The quadratic spin equality $\pmb{x}^T\pmb{J}\pmb{x}=c$ where $\pmb{J}$ is a hermitian matrix and $c$ is a constant, it can be written as
\begin{align}
&\pmb{x}^T\pmb{J}\pmb{x}=c \Rightarrow \pmb{x}^T\pmb{J}\pmb{x}-c=0 \nonumber \\
&\Rightarrow \pmb{x}^T\pmb{J}\pmb{x}-\frac{c}{N}\pmb{x}^T\pmb{x}=0 \nonumber \\
&\Rightarrow \pmb{x}^T \left(\pmb{J}-\frac{c}{N}\pmb{I} \right)\pmb{x}=0  \nonumber \\
&\Rightarrow \pmb{b}^T \pmb{J}'\pmb{b}+C=0, \label{ina}
\end{align}
where $\pmb{J}'=4(\Re(\pmb{J})-\frac{c}{N}\pmb{I})-4\diag((\Re(\pmb{J})-\frac{c}{N}\pmb{I})\pmb{1})$, $C=\sum_{i,j}J_{i,j}-c$,  and the expression in \eqref{ina} uses Lemma \ref{lma}.

\section{Proof of Proposition \ref{prop}}\label{prop_prf}
Let $x^*_i = e^{j\chi^*_i},$ with $\chi^*_i \in \{0,\pi\}$. The average channel gain with $k$ phase shifts equal to $+1$ can be written as
\begin{align}
H_k &= \mathbb{E}\left(\left|\sum_{i=1}^N |h_i| |g_i| e^{j(\chi^*_i+\theta_i+\phi_i)}\right|^2\right)\nonumber\\
&=\mathbb{E}\Bigg(\sum_{i=1}^N |h_i|^2 |g_i|^2 + \sum_{i=1}^N \sum_{m\neq i}^N |h_i| |g_i| |h_m| |g_m|\nonumber\\
&\hspace{3.5cm}\times e^{j(\chi^*_i+\theta_i+\phi_i-\chi^*_m-\theta_m-\phi_m)}\Bigg)\nonumber\\
&= N + \frac{\pi^2}{16} \mathbb{E}\Bigg(\sum_{i=1}^N \sum_{m\neq i}^N e^{j(\chi^*_i+\theta_i+\phi_i-\chi^*_m-\theta_m-\phi_m)}\Bigg),\label{crossterms}
\end{align}
which follows from $\mathbb{E}(|h_i|^2) = \mathbb{E}(|g_i|^2) = 1$ and $\mathbb{E}(|h_i|) = \mathbb{E}(|g_i|) = \sqrt{\pi}/2$. Note that $\theta_i + \phi_i$ is a uniform random variable in $[0,2\pi)$ as $\theta_i$ and $\phi_i$ are independent and uniform random variables in $[0,2\pi)$. In the conventional beamforming case, the quantization error $\chi^*_i - \theta_i - \phi_i$ is uniformly distributed in $[-\pi/2, \pi/2)$ \cite{ZHANG}, as each discrete phase shift can be mapped to the nearest continuous one. In other words, all phases are shifted to the first and fourth quadrant.
  
However, due to the limitation imposed by the IM-based technique, depending on $k$, the phases can be shifted or remain within the second and third quadrant {\it i.e.,} the quantization error is uniformly distributed in $[\pi/2, 3\pi/2)$. Therefore, the optimal phase shift design is to maximize the difference between the number of phases between these two regions. For example, assume there are two phases within $[\pi/2, 3\pi/2)$ and $N-2$ phases within $[-\pi/2, \pi/2)$. Then, if $k=1$, the optimal design is to shift one phase from within $[\pi/2, 3\pi/2)$ and $N-2$ phases from within $[-\pi/2, \pi/2)$, resulting in $N-1$ phases in $[\pi/2, 3\pi/2)$ and one phase in $[-\pi/2, \pi/2)$. If $k=2$, the optimal design is to shift the $N-2$ phases from within $[-\pi/2, \pi/2)$, resulting in $N$ phases in $[\pi/2, 3\pi/2)$. Hence, the terms in \eqref{crossterms} correspond to cross terms within the same region and cross terms between the two regions. Assuming there are $n$ phases within $[\pi/2, 3\pi/2)$, the optimal design will result in $(N-|k-n|)$ phases in $[\pi/2, 3\pi/2)$ and $|k-n|$ phases in $[-\pi/2, \pi/2)$. Therefore, we have
\begin{align}
&\mathbb{E}\Bigg(\sum_{i=1}^N \sum_{m\neq i}^N e^{j(\chi^*_i+\theta_i+\phi_i-\chi^*_m-\theta_m-\phi_m)}\Bigg)\nonumber\\
&= \frac{4}{\pi^2}\mathbb{E}_n\big(|k-n|(|k-n|-1)\nonumber\\
&\qquad\qquad+ (N-|k-n|)(N-|k-n|-1)\nonumber\\
&\qquad\qquad- 2|k-n|(N-|k-n|)\big),
\end{align}
where the first two terms refer to the number of pairs within the same region with $\mathbb{E}(e^{j(\chi^*_i+\theta_i+\phi_i-\chi^*_m-\theta_m-\phi_m)}) = 4/\pi^2$, and the last term to the number of pairs between the regions with $\mathbb{E}(e^{j(\chi^*_i+\theta_i+\phi_i-\chi^*_m-\theta_m-\phi_m)}) = -4/\pi^2$. Note that there are $\lfloor N/2\rfloor+1$ (due to symmetry) distinct values to consider for $n$. Moreover, since $\theta_n + \phi_n$ is uniform in $[0,2\pi)$, the probability of $n$ phases appearing is $(1/2)^N \binom{N}{n}$. Therefore,
\begin{align}
H_k = N + \frac{\pi^2}{16} \frac{4}{\pi^2} \sum_{n=0}^{\lfloor N/2\rfloor}& 2\frac{1}{2^N} \binom{N}{n}\bigg[|k-n|(|k-n|-1)\nonumber\\
&+ (N-|k-n|)(N-|k-n|-1)\nonumber\\
&- 2|k-n|(N-|k-n|)\bigg].
\end{align}
Observe that each term is considered twice. This is valid for odd $N$, whereas for even $N$, the $N/2$ term needs to be considered just once. After several algebraic manipulations, the final expressions are derived.

\begin{IEEEbiography}[{\includegraphics[width=1in,height=1.25in,clip,keepaspectratio]{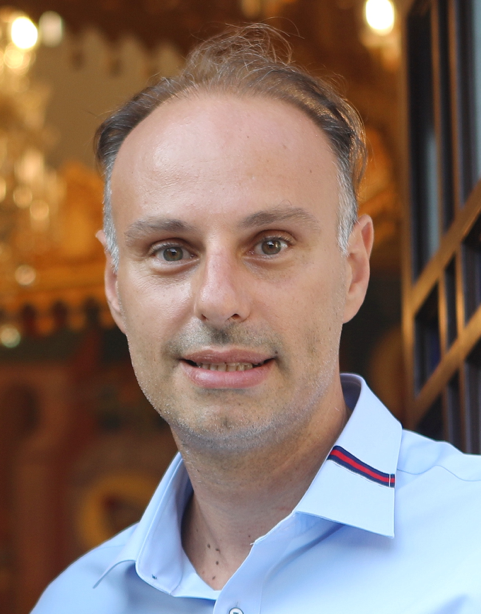}}]{Ioannis Krikidis} (F?19) received the diploma in Computer Engineering from the Computer Engineering and Informatics Department (CEID) of the University of Patras, Greece, in 2000, and the M.Sc and Ph.D degrees from \'Ecole Nationale Sup\'erieure des T\'el\'ecommunications (ENST), Paris, France, in 2001 and 2005, respectively, all in Electrical Engineering. From 2006 to 2007 he worked, as a Post-Doctoral researcher, with ENST, Paris, France, and from 2007 to 2010 he was a Research Fellow in the School of Engineering and Electronics at the University of Edinburgh, Edinburgh, UK.
	
He is currently a Professor at the Department of Electrical and Computer Engineering, University of Cyprus, Nicosia, Cyprus. His current research interests include wireless communications, quantum computing, 6G communication systems, wireless powered communications, and intelligent reflecting surfaces. Dr. Krikidis serves as an Associate Editor for IEEE Transactions on Wireless Communications, and Editor in Chief for Frontiers in Communications and Networks. He was the recipient of the Young Researcher Award from the Research Promotion Foundation, Cyprus, in 2013, and the recipient of the IEEEComSoc Best Young Professional Award in Academia, 2016, and IEEE Signal Processing Letters best paper award 2019. He has been recognized by the Web of Science as a Highly Cited Researcher for 2017-2021. He has received the prestigious ERC Consolidator Grant for his work on wireless powered communications.
\end{IEEEbiography}

\begin{IEEEbiography}[{\includegraphics[width=1in,height=1.25in,clip,keepaspectratio]{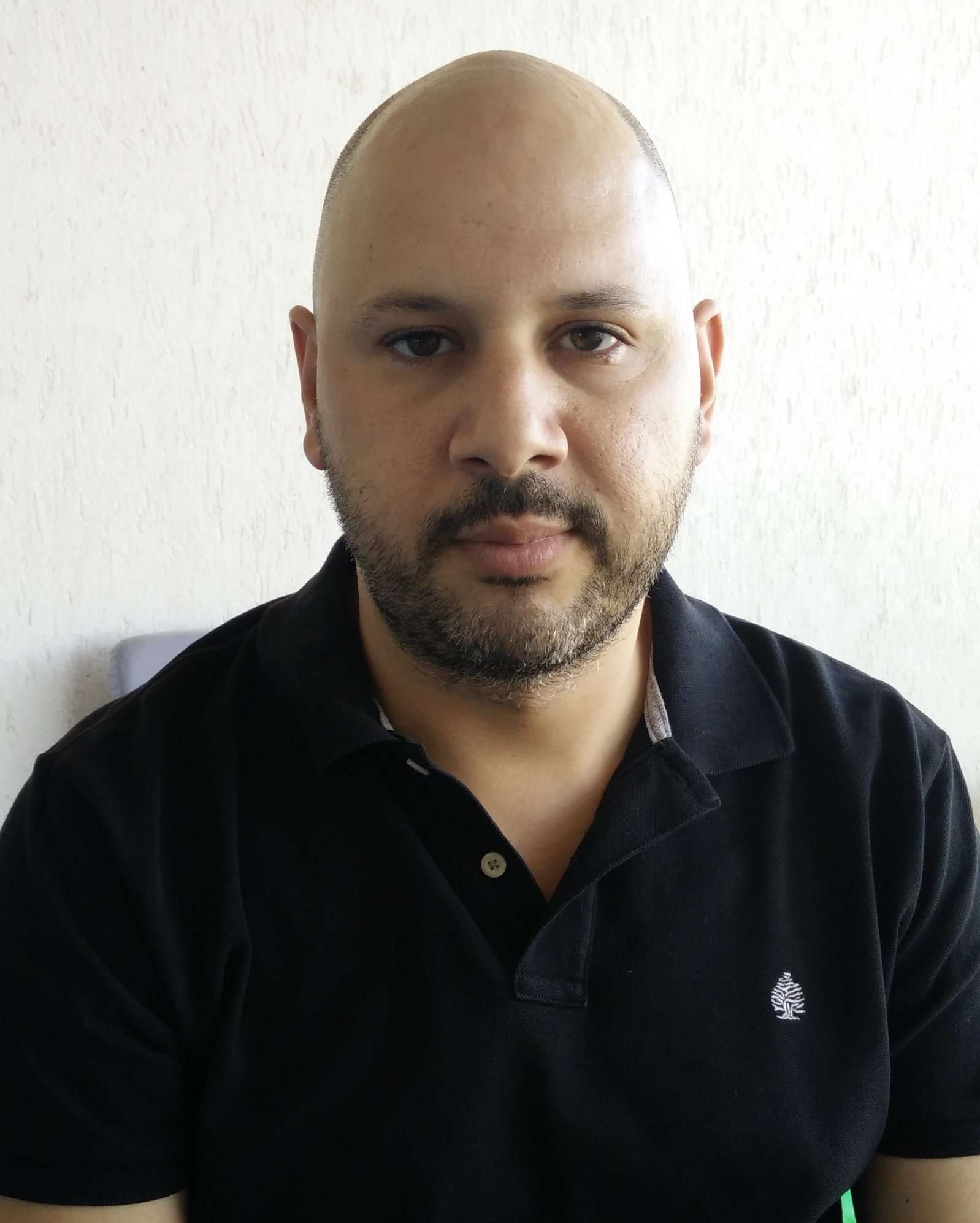}}]{Constantinos Psomas} (M'15-SM'19) holds a B.Sc. in Computer Science and Mathematics (First Class Honours) from Royal Holloway, University of London, an M.Sc. in Applicable Mathematics from the London School of Economics and a Ph.D. in Mathematics from The Open University, UK. He is an Assistant Professor at the Department of Computer Science and Engineering, European University Cyprus. From 2011 to 2014, he was a Postdoctoral Research Fellow with the Department of Electrical Engineering, Computer Engineering and Informatics, Cyprus University of Technology and, from 2014 to 2025, he was a Research Fellow with the Department of Electrical and Computer Engineering, University of Cyprus. Dr. Psomas serves as an Associate Editor for the IEEE Transactions on Communications. His current research interests include wireless powered communications, fluid/movable antennas and reconfigurable intelligent surfaces.
\end{IEEEbiography}

\begin{IEEEbiography}[{\includegraphics[width=1in,height=1.25in,clip,keepaspectratio]{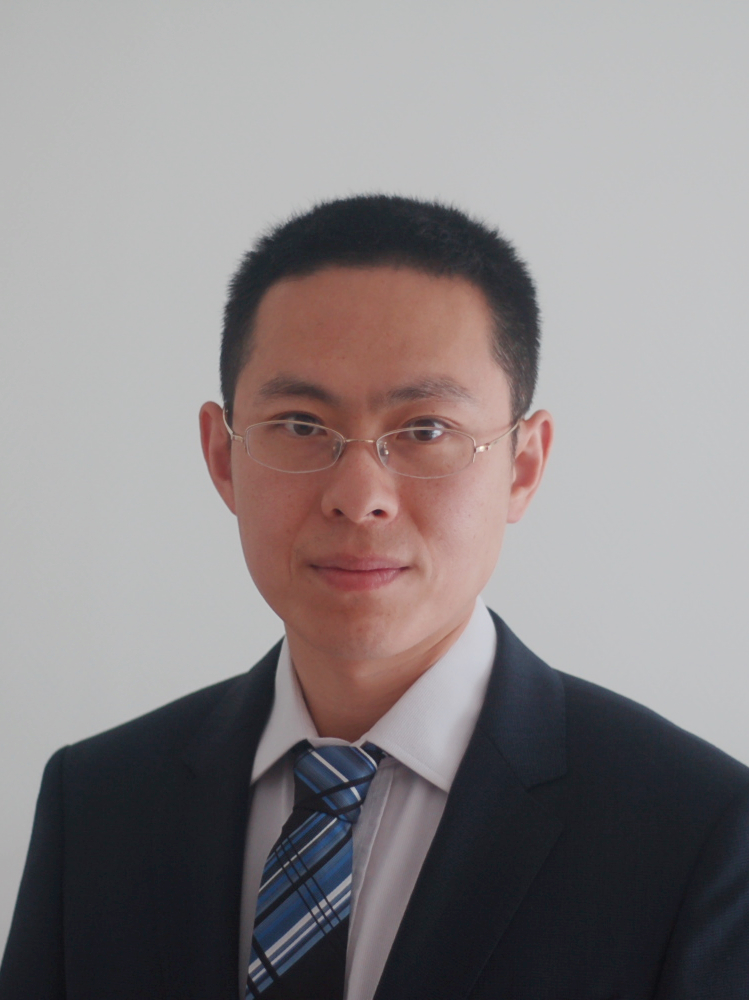}}]
	{Gan Zheng} (S'05-M'09-SM'12-F'21)  received the BEng and the MEng from Tianjin University, Tianjin, China, in 2002 and 2004, respectively, in Electronic and Information Engineering, and the PhD degree in Electrical and Electronic Engineering from The University of Hong Kong in 2008. He is currently Professor in Connected Systems in the School of Engineering, University of Warwick, UK. His research interests include machine learning and quantum computing for wireless communications, reconfigurable intelligent surface, fluid antennas and edge computing. He is the first recipient for the 2013 IEEE Signal Processing Letters Best Paper Award, and he also received 2015 GLOBECOM Best Paper Award, and 2018 IEEE Technical Committee on Green Communications \& Computing Best Paper Award. He was listed as a Highly Cited Researcher by Thomson Reuters/Clarivate Analytics in 2019. He currently serves as an Associate Editor for IEEE Wireless Communications Letters and IEEE Transactions on Communications.
\end{IEEEbiography}

\end{document}